\documentclass[final]{siamltex}
\usepackage{amsfonts}
\usepackage{graphics}
\usepackage{caption2}
\usepackage{amssymb}
\usepackage{psfrag}
\usepackage{amsmath}
\usepackage{graphicx}
\usepackage{multirow}
\usepackage{cite}
\usepackage{subfigure}
\usepackage[normalem]{ulem}
\usepackage{color}

\allowdisplaybreaks[4]

\def\x{{\bf   x}}

\def\Q{{\bf   Q}}

\def\({\left(}
\def\[{\left[}
\def\){\right)}
\def\]{\right]}

\def\grad{\nabla }

\newtheorem{rem}{Remark}

\newtheorem{lem}{Lemma}
\newtheorem{thm}{Theorem}

\numberwithin{equation}{section}
\numberwithin{thm}{section}
\numberwithin{lem}{section}
\numberwithin{rem}{section}

\begin{document}
\title{Stabilized   energy factorization   approach for Allen--Cahn equation with logarithmic Flory--Huggins   potential}

\author{
Xiuhua Wang\thanks{School of Mathematics and Statistics, Hubei Engineering  University, Xiaogan 432000, Hubei, China.}
\and Jisheng Kou\thanks{Corresponding author.  School of Civil Engineering, Shaoxing University, Shaoxing 312000, Zhejiang, China. Email: {\tt jishengkou@163.com}.} 
\and Jianchao Cai\thanks{Corresponding author. Institute of Geophysics and Geomatics, China University of Geosciences, Wuhan 430074,  Hubei, China.
 Email: {\tt caijc@cug.edu.cn.}}
  }
 \maketitle

\begin{abstract}
The Allen--Cahn equation is one of fundamental equations of phase-field models, while the logarithmic Flory--Huggins   potential is one of the most useful  energy potentials in various phase-field models.  In this paper, we consider numerical schemes for solving the Allen--Cahn equation with logarithmic Flory--Huggins potential. The main challenge is how to design efficient numerical schemes that preserve the maximum principle and energy dissipation law due to the strong nonlinearity of the energy potential function.  We propose  a novel energy factorization  approach with the stability technique, which is called stabilized energy factorization approach,  to deal with the Flory--Huggins potential.  One advantage of the proposed approach is that all nonlinear terms can be treated semi-implicitly and the resultant numerical scheme is purely  linear and easy to implement.  Moreover, the discrete maximum principle and unconditional energy stability of the proposed scheme are rigorously proved using the discrete variational principle. Numerical results are presented to demonstrate the stability and effectiveness of the proposed scheme.

\end{abstract}
 
\section{Introduction}   

 It is known that the logarithmic Flory--Huggins potential \cite{Cahn1977LFH,Yang2019IEQ,Wells2006FLP,chen2017LP,Gomez2011LP,Copetti1992FLP,Gomez2008LP,Kastner2016LFP,Wodo2011FLP,Wells2006FLP} and double well potential \cite{Kay2009PhaseField,Kim2012PhaseField,Han2015DWconvex,Zhao2016DWconvex,Liu2015DWconvex,Guillen2013DWgeneralAV,chen2015DWstab,Ma2017DWstab,Xu2006Stab,shen2015DWstab} are  the most useful energy potentials in phase field models. Moreover, the Flory--Huggins energy potential is   more  realistic  than the double well potential  from the viewpoint of physics\cite{Cahn1977LFH,chen2017LP}.  However, compared to the double well potential,  there are  more scarce    works of numerical algorithm developments and numerical analysis devoted to the Flory--Huggins potential  \cite{Yang2019IEQ}.  There are two possible reasons: one  is the difficulty resulting from   the strong nonlinearity of this potential function, and the other is  the singularity  in the situations that   the phase variable approaches the limits 0 and 1. So it becomes a very challenging  issue to design  efficient  algorithms  for phase-field models with the Flory--Huggins potential.  The Allen--Cahn equation, as well as the Cahn--Hilliard equation \cite{Cahn1958CH,Cahn1977LFH}, is the fundamental equation of phase-field models,  so it has become an attractive research topic in mathematical analysis \cite{Evans1992ACmaximumprinciple} and numerical simulation \cite{Feng2003AC,Feng2013AC}.  The maximum principle and    energy dissipation law are two intrinsic  properties of   the   Allen--Cahn equation        \cite{Evans1992ACmaximumprinciple}.      Numerical schemes preserving these    key properties are highly preferred for  solving this equation.  However, it is still   challenging   to design numerical schemes that mimic these properties.   A brief review on the developments of such schemes in the literature will be stated as below. The  main goal of this work is to propose \emph{linear, unconditionally energy stable} and  \emph{maximum principle preserving} numerical schemes for the Allen--Cahn equation with the Flory--Huggins   potential.

 The traditional  explicit and fully implicit time marching schemes for bulk free energy functions have been shown to suffer from very severe constraints  on the time step size \cite{Boyer2011CH,Feng2003AC,Copetti1992FLP}. Therefore, it becomes a major issue to design the   semi-implicit  energy stable numerical schemes that remove the constraints on the time step size.   
  A popular  approach used for designing  such schemes  is the convex splitting approach \cite{Elliott1993ConvexSpliting,Eyre1998ConvexSplitting,shen2015SIAM,Wise2009Convex,Baskaran2013convexsplitting,Han2015DWconvex,Zhao2016DWconvex,Liu2015DWconvex,Guillen2013DWgeneralAV,shen2012convex,Li2017Convex,Zhang2010DWconvex}, which   produces  unconditionally energy stable schemes by the use of implicit treatment for the convex terms of the energy functions    and explicit treatment for the concave terms. Besides various phase-field models,  it has also been successfully applied to deal with the  Helmholtz free energy    based on a realistic  equation of state   in recent years \cite{qiaosun2014,fan2017componentwise,kouandsun2017modeling,kousun2018Flash,Peng2017convexsplitting,kou2017nonisothermal,kou2018nonisothermal}.  However, it usually results in the nonlinear schemes; for instance, the convex splitting approach  for   the logarithmic Flory--Huggins   potential  leads to the scheme  involving the implicit    logarithmic function.  As a consequence, practical implementation of such schemes  requires nonlinear iterative solvers and computational cost may be expensive accordingly. 
The other  commonly used conventional  approach  for constructing    linear numerical schemes is  the linear stabilization approach \cite{chen2015DWstab,Ma2017DWstab,Xu2006Stab,Yang2009stab,shen2015DWstab,Kim2012PhaseField,kousun2015SISC},  which simply treats  all nonlinear terms by the fully explicit way and introduces a linear stabilization  term    to remove the time step constraint. We observe that the  stabilization   approach is effective for the double well potential, but it works not well for  the  logarithmic  potential probably because of more complicate nonlinearity of it.  In \cite{kou2017compositional,kouandsun2017modeling},  a nonlinear stabilization approach has been proposed for  the Peng-Robinson equation of state  (PR-EOS) \cite{Peng1976EOS} that is   one of the most  useful   tools in petroleum industry and chemical engineering.  For the  logarithmic Flory--Huggins   potential,  the instability will occur and numerical results may be out of normal range  when the energy parameter takes a large value. Inspired by the approach in \cite{kou2017compositional,kouandsun2017modeling}, we will  incorporate the nonlinear stabilization term for  the  logarithmic Flory--Huggins   potential to    ensure the symmetric   positive definiteness and  discrete maximum principle of the resultant scheme in the case with a large energy parameter.  

In recent years,   the novel auxiliary variable approaches have been developed and successfully applied to design linear numerical schemes for various diffuse interface models \cite{Yang2019IEQ,Yang2017IEQ,Yang2017IEQ2,Zhao2017DW_IEQ,Li2017IEQ,kou2018SAV,Yang2019gPAV,Shen2018SAV,Shen2018SAVConvergence,Zhu2018SAV,Zhu2019SAV}.  The first approach is the so-called invariant energy quadratization (IEQ) approach   \cite{Yang2019IEQ,Yang2017IEQ,Yang2017IEQ2,Zhao2017DW_IEQ} that   has been successfully  applied  to devise efficient, linear schemes for various  phase-field models intensively in recent years.  The basic  idea of   IEQ        is to   define    a set of auxiliary variables  and then transform the  original  free energy function  into a quadratic form.   The second approach is the scalar auxiliary variable (SAV) approach  proposed in \cite{Shen2018SAV}, which differs from IEQ that it uses a scalar auxiliary variable instead of the space-dependent auxiliary variable. The convergence and error estimates for the SAV schemes are studied in \cite{Shen2018SAVConvergence}.  Very recently, a generalized positive auxiliary variable, termed gPAV, is proposed in \cite{Yang2019gPAV} inspired by    SAV and IEQ.      Numerical schemes developed   by the  auxiliary variable  approaches are linear   and easy to implement.   As a result,  such   approaches  have rapidly become the   useful and successful     tools   for simulating  a variety of diffuse interface models.  The auxiliary variable approaches use the    transformed    energies that are equivalent to the original  energies at the continuous  level, but   there are   the discrete errors   between   the original  energies and   transformed   energies at the time-discrete level \cite{Yang2019IEQ},  
thus  the resultant    schemes   may  not preserve  the original energy dissipation  law       although the     transformed energy dissipation can be proved.

More recently, a novel energy factorization (EF) approach is proposed  in   \cite{kou2019EF} to design linear,  efficient numerical schemes for the diffuse interface model with PR-EOS.
The basic  idea of EF  is to factorize an energy function/term   into   a product of several factors, which can be treated in the energy difference  by the use of individual  properties of each factor.  Compared with the convex splitting approach, the EF approach can produce the linear  semi-implicit    schemes. It is different from the IEQ/SAV approach that the  EF approach never introduces any new independent energy variable, and thus the resultant schemes can preserve  the original  energy dissipation law.  The EF approach has been successfully applied to  design  the   semi-implicit linear schemes for the PR-EOS model \cite{kou2019EF}.   In this paper, we will propose   the stabilized EF approach  to deal with  the logarithmic Flory--Huggins energy potential,   and from this,  we will propose  a linear  semi-implicit discrete scheme    inheriting  the original energy dissipation law.  The proposed scheme is efficient and very easy to implement.

For the phase-field models, the phase variables usually comply with the specific  maximum principles in terms of their physical meanings. It is   known that the   Allen-Cahn equation    satisfies the maximum principle \cite{Evans1992ACmaximumprinciple}. A discrete scheme preserving the maximum principle is highly preferred for solving such equations since it can eliminate spurious numerical solutions such that it can  not only   ensure   the physical reasonability  of numerical results  but also    improve the long-time stability of   numerical simulation substantially.  However, the efforts regarding  the maximum principle preserving discrete schemes and numerical analysis are even more scarce partially  because such schemes are actually very challenging due to   particularities of  involved  spatial and temporal discretization \cite{chen2017LP}.
In \cite{Tang2016AC}, the commonly used numerical method for the Allen-Cahn equation with the double well potential was proved to preserve the maximum principle by the matrix analysis approach.   For the  logarithmic Flory--Huggins energy potential, the implicit Euler scheme  was analyzed   in \cite{Copetti1992FLP} under the restricted  constraint on the time step size due to   the fully implicit treatment for all energy terms, and recently,  the discrete maximum principle of the semi-implicit convex splitting  scheme was analyzed in \cite{chen2017LP}. More recently,   the maximum principle of a linear numerical scheme for  the PR-EOS diffuse interface model  has been proved in \cite{kou2019EF}. However, there are no results regarding the linear semi-implicit scheme with the maximum principle for the phase-field models with  the  logarithmic Flory--Huggins energy potential. This gap will be filled in this paper. 
 The discrete maximum principle  of the proposed   scheme    will be rigorously proved  using the discrete variational principle    \cite{kou2019EF}. Appropriate  stability strategy 
plays a  crucial  role in preserving   the discrete maximum principle.  
 
The cell-centered finite difference (CCFD) method \cite{Tryggvason2011book} is applied  as the spatial discretization method. The CCFD method can be equivalent to a special mixed finite element method  with quadrature rules \cite{arbogast1997mixed}  and  has been applied for the phase field models \cite{Furihata2001FDM,Hu2009FDM,Khiari2007FDM,Wise2009Convex,Wise2010CahnHilliardHeleShaw,shen2016JCP}.  The discrete variational principle will be used to carry out theoretical analysis of the proposed scheme.

The new aspects of the current work   are listed as follows:
(1) the stabilized  energy factorization  approach    to treat the logarithmic Flory--Huggins   potential; (2) the purely  linear numerical scheme inheriting the original  energy dissipation law; (3) the discrete maximum principle of the proposed scheme.

The rest of this paper is structured as   follows. In Section \ref{secModelEqn}, we describe the Allen--Cahn equation with Flory--Huggins potential.  In Section \ref{secDiscreteSpaces}, we   introduce the discrete function spaces and discrete operators based on CCFD.   In Sections \ref{secSchemeFH},  we   propose the stabilized   energy factorization   approach to deal  with the logarithmic Flory--Huggins   potential,   and then present  the  fully discrete scheme. In Section \ref{secanalysis},   we carry out theoretical analysis, including the well posedness of the discrete solutions,   discrete maximum principle and unconditional energy stability. 
 In Section \ref{secNumer},  numerical results are presented  to demonstrate the effectiveness and stability of  the proposed   scheme.   Finally,   the concluding remarks are given in Section \ref{secConclusions}.

\section{Allen--Cahn equation with Flory--Huggins potential}\label{secModelEqn}

This paper is concerned with the development of efficient numerical schemes for  the following Allen--Cahn equation \cite{Allen1979ACeq}
 \begin{equation}\label{eqAC}
 \frac{ \partial\phi(\x,t) }{\partial t}
 -\epsilon^2\Delta\phi(\x,t) +f (\phi(\x,t))=0,~~~\x\in\Omega,
 \end{equation}
 where $t$ is the time, $\Omega$ denotes the bounded domain with smooth boundaries and $\epsilon$  is a  positive constant  measuring the interfacial width.  Here, $\phi$ represents  the phase variable and   $f(\phi)=F'(\phi)$ is the chemical potential, where   $F(\phi)$ is the Helmholtz free energy. Meanwhile, the   equation \eqref{eqAC} is   subject to specified  initial condition and    homogeneous Neumann/Dirichlet or periodic   boundary conditions.
 
  In this paper, we consider  the logarithmic Flory--Huggins energy potential \cite{Cahn1977LFH,Yang2019IEQ,Wells2006FLP}
 \begin{equation}\label{eqLFH}
 F(\phi)=\phi\ln(\phi)+(1-\phi)\ln(1-\phi)+\theta\(\phi- \phi^2\),~~0<\phi<1,
\end{equation}
where $\theta>2$ is the energy parameter.  It is noticed in  \cite{Wells2006FLP} and also shown in Figure \ref{LogarithmicFloryHugginsProfiles} that the choice of $\theta>2$ is necessary for $F(\phi)$ since in this case it  has two wells   and admits two phases; otherwise,  for
$\theta\leq2$,   only  a single well exists as well as a single phase. It is clear that the maximum principle, i.e. $0<\phi<1$, is crucial to ensure that $F(\phi)$ and $f(\phi)$ are well defined in  both mathematics and  physics. 

 The equation \eqref{eqAC} follows the  energy dissipation law; in fact, 
we define the energy function 
\begin{equation}\label{eqACEnergy}
E(\phi)=\int_\Omega \(F(\phi)+\frac{1}{2}\epsilon^2|\grad\phi|^2\)d\x,
\end{equation}
then  $E(\phi)$ is decreasing with time as follows
\begin{equation}\label{eqACEnergyDecay}
\frac{\partial }{\partial t}E(\phi)=-\int_\Omega\left|\frac{\partial \phi}{\partial t}\right|^2d\x\leq0.
\end{equation}

  \begin{figure}
           \centering \subfigure[Energy potential]{
            \begin{minipage}[b]{0.38\textwidth}
            \centering
             \includegraphics[width=0.95\textwidth,height=0.85\textwidth]{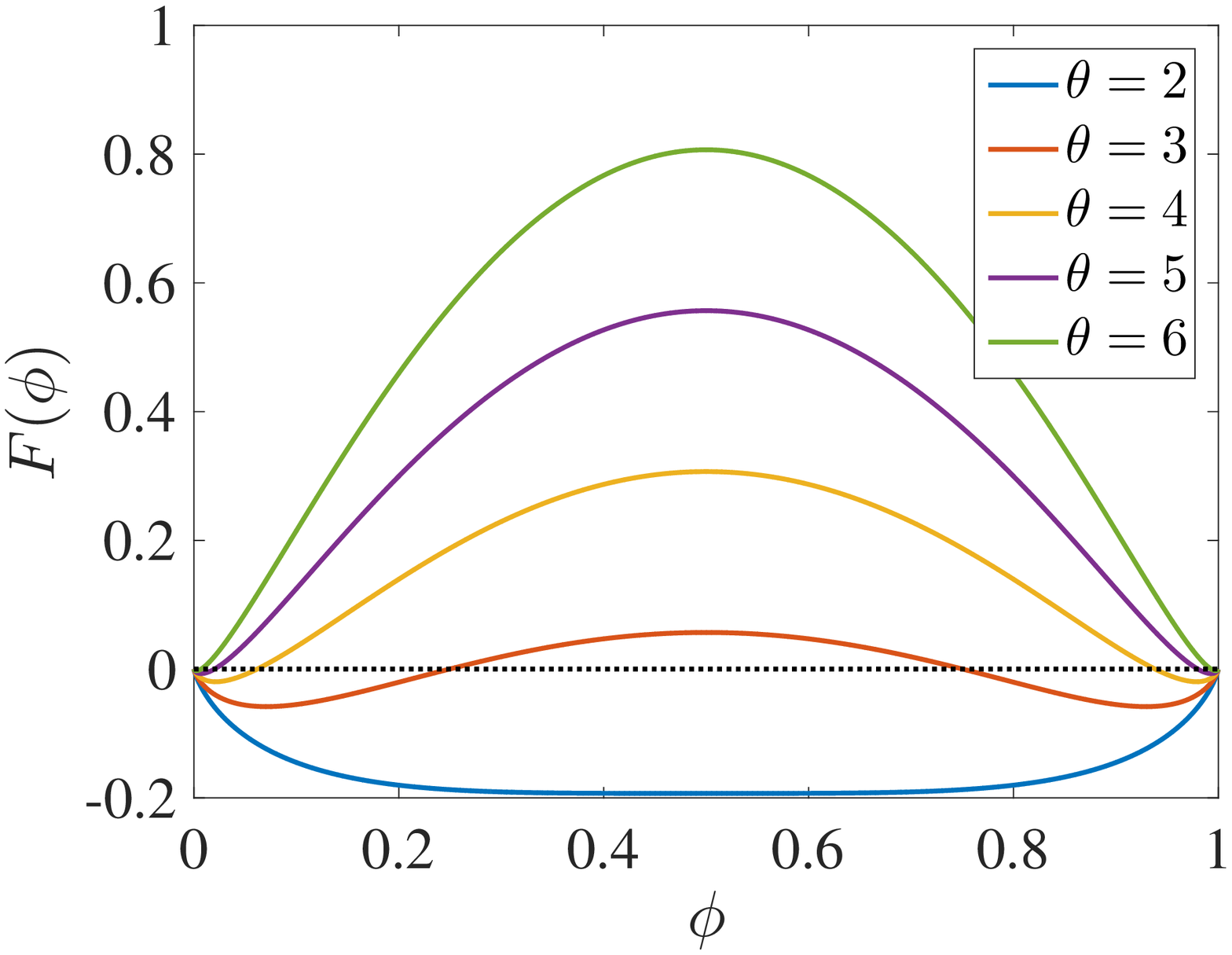}
             \label{LogarithmicFloryHugginsFunctionProfile}
            \end{minipage}
            }
            \centering \subfigure[Chemical potential]{
            \begin{minipage}[b]{0.38\textwidth}
            \centering
             \includegraphics[width=0.95\textwidth,height=0.85\textwidth]{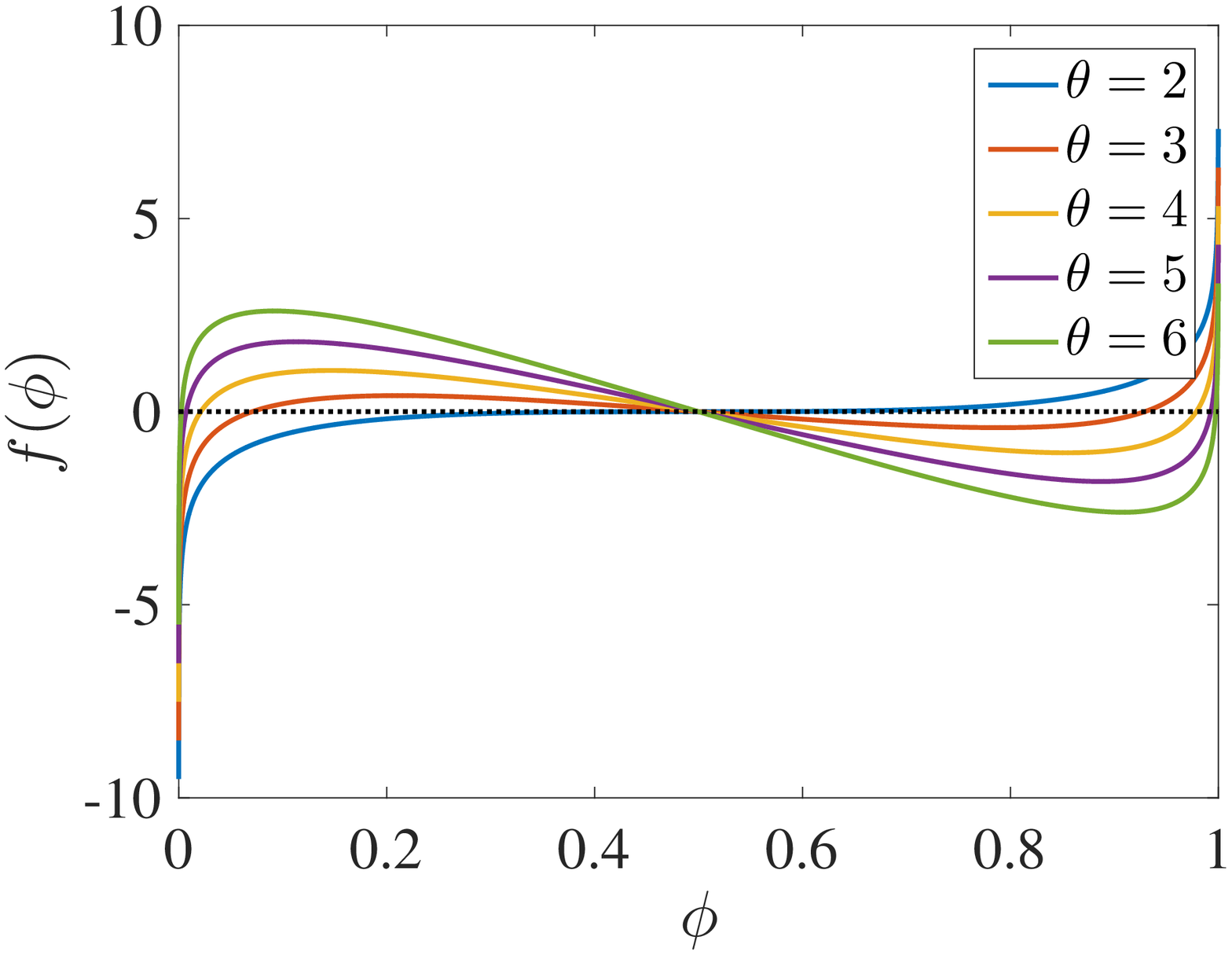}
             \label{LogarithmicFloryHugginsChemicalPotentialProfile}
            \end{minipage}
            }
            \caption{Logarithmic Flory--Huggins energy potential and   chemical potential.}
            \label{LogarithmicFloryHugginsProfiles}
 \end{figure}

\section{Discrete function spaces and discrete operators}\label{secDiscreteSpaces}
In this section, we describe the discrete function spaces,  discrete operators and discrete variational principle \cite{Tryggvason2011book,Furihata2001FDM,Hu2009FDM,Khiari2007FDM,Wise2009Convex,Wise2010CahnHilliardHeleShaw,shen2016JCP,kou2019EF,kou2018SAV}, which are based  on the cell-centered finite difference (CCFD) method \cite{Tryggvason2011book,arbogast1997mixed}.
Here, we consider the two-dimensional case only, but   the forms of the  three-dimensional case are similar.
Let  $\Omega=[l_x^a,l_x^b]\times[l_y^a,l_y^b]$, where $l_x^a<l_x^b$ and $l_y^a<l_y^b$.   For simplicity,    the uniform division   is used and      the mesh size is denoted by  $h=\frac{l_x^b-l_x^a}{N}=\frac{l_y^b-l_y^a}{M}$, where $N$ and $M$ are  positive  integers. 
The grid points are denoted by  $x_i=l_x^a+ih$ and $y_j=l_y^a+jh$, and furthermore,  the intermediate points  are denoted by $x_{i+\frac{1}{2}} =x_i+\frac{1}{2}h$ and $y_{j+\frac{1}{2}} =y_j+\frac{1}{2}h$.   

The discrete function spaces are defined as follows
  \begin{equation*}
\mathcal{C}_h=\big\{\phi: (x_{i+\frac{1}{2}},y_{j+\frac{1}{2}})\mapsto\mathbb{R},~~0\leq i\leq N-1, ~0\leq j\leq M-1\big\},
  \end{equation*}
\begin{equation*}
\mathcal{U}_h=\big\{u: (x_{i},y_{j+\frac{1}{2}})\mapsto\mathbb{R},~~0\leq i\leq N, ~0\leq j\leq M-1\big\},
  \end{equation*}
\begin{equation*}
\mathcal{V}_h=\big\{v: (x_{i+\frac{1}{2}},y_{j})\mapsto\mathbb{R},~~0\leq i\leq N-1, ~0\leq j\leq M\big\}.
  \end{equation*}
Here, we use the identification  
 $\phi_{i+\frac{1}{2},j+\frac{1}{2}}=\phi(x_{i+\frac{1}{2}},y_{j+\frac{1}{2}})$, etc.
 The homogeneous  Neumann boundary condition is considered in this paper, and the case with the homogeneous  Dirichlet or periodic   boundary condition can be formulated  and analyzed analogously. We   introduce  the subsets of  $\mathcal{V}_u$ and $\mathcal{V}_v$ involving   the homogeneous  Neumann boundary condition  as 
\begin{equation}
\mathcal{U}^0_h=\big\{u\in\mathcal{U}_h| ~u_{0,j+\frac{1}{2}}=u_{N,j+\frac{1}{2}}=0,~~0\leq j\leq M-1\big\},
  \end{equation}
\begin{equation}
\mathcal{V}^0_h=\big\{v\in\mathcal{V}_h|~v_{i+\frac{1}{2},0}=v_{i+\frac{1}{2},M}=0,~~0\leq i\leq N-1\big\}.
  \end{equation}
For $\phi\in \mathcal{C}_h$, the discrete gradient operator is defined as $\grad_h = [\grad_{h,x},\grad_{h,y}]^T$, where $\grad_{h,x}\phi\in \mathcal{U}_h^0$ and $\grad_{h,y}\phi\in \mathcal{V}_h^0$ are defined as follows
\begin{subequations}\label{eqFullyDiscreteGradofMolarDens}
\begin{equation}\label{eqFullyDiscreteGradofMolarDens01}
\grad_{h,x}\phi_{i,j+\frac{1}{2}}=\frac{\phi_{i+\frac{1}{2},j+\frac{1}{2}}-\phi_{i-\frac{1}{2},j+\frac{1}{2}}}{h},~~1\leq i\leq N-1, ~0\leq j\leq M-1,
  \end{equation}
\begin{equation}\label{eqFullyDiscreteGradofMolarDens02}
\grad_{h,y}\phi_{i+\frac{1}{2},j}= 
\frac{\phi_{i+\frac{1}{2},j+\frac{1}{2}}-\phi_{i+\frac{1}{2},j-\frac{1}{2}}}{h},~~0\leq i\leq N-1, ~1\leq j\leq M-1.
  \end{equation}
 \end{subequations}
For $u\in \mathcal{U}_h $ and $v\in \mathcal{V}_h $, the  discrete divergence    operator is defined as $\grad_h\cdot[u,v]^T=D_xu+D_yv$, where    $D_x$ and $D_y$ are defined   as follows
\begin{equation}\label{eqFullyDiscreteGradofMolarDens01}
D_{x} u_{i+\frac{1}{2},j+\frac{1}{2}}=\frac{u_{i+1,j+\frac{1}{2}}-u_{i,j+\frac{1}{2}}}{h},
~~~D_y v_{i+\frac{1}{2},j+\frac{1}{2}}=\frac{v_{i+\frac{1}{2},j+1}-v_{i+\frac{1}{2},j}}{h},
  \end{equation}
  where $0\leq i\leq N-1, ~0\leq j\leq M-1$.
  The discrete Laplace operator is defined as $\Delta_h=\grad_h\cdot\grad_h$. 
  
    We define the following  discrete inner-products:
\begin{equation*}
\( \phi,\phi'\)_h=h^2\sum_{i=0}^{N-1}\sum_{j=0}^{M-1}\phi_{i+\frac{1}{2},j+\frac{1}{2}}\phi'_{i+\frac{1}{2},j+\frac{1}{2}},
~~~\phi,\phi' \in\mathcal{C}_h,
\end{equation*}
\begin{equation*}
\(u,u'\)_h=h^2\sum_{i=1}^{N-1}\sum_{j=0}^{M-1}u_{i,j+\frac{1}{2}}u'_{i,j+\frac{1}{2}},
~~~u,u'\in \mathcal{U}_h^0,
\end{equation*}
\begin{equation*}
\( v,v'\)_h=h^2\sum_{i=0}^{N-1}\sum_{j=1}^{M-1}v_{i+\frac{1}{2},j}v'_{i+\frac{1}{2},j},~~~v,v'\in \mathcal{V}_h^0.
\end{equation*}
The discrete norms for $\phi \in\mathcal{C}_h$, $u \in\mathcal{U}_h^0$ and $v \in\mathcal{V}_h^0$ are denoted as
\begin{equation*}
\|\phi\|_h =\( \phi,\phi\)_h^{1/2},~~\|u\|_h =\(u,u\)_h^{1/2},~~\|v\|_h =\(v,v\)_h^{1/2}, 
\end{equation*}
\begin{equation*}
\|\grad_h\phi\|_h^2=\|\grad_{h,x}\phi\|_h^2+\|\grad_{h,y}\phi\|_h^2.
\end{equation*}
 The following discrete variational  formulas are obtained  by direct calculations \cite{Wise2009Convex,Wise2010CahnHilliardHeleShaw,shen2016JCP,kou2019EF,kou2018SAV}
\begin{equation}\label{eqFullyDiscreteVariationalPrinciples01}
\( u,\grad_{h,x}\phi\)_h=-\( D_x u,\phi\)_h, ~~u\in \mathcal{U}_h^0,~\phi\in \mathcal{C}_h,
\end{equation}
\begin{equation}\label{eqFullyDiscreteVariationalPrinciples02}
\( v,\grad_{h,y}\phi\)_h=-\(D_y v,\phi \)_h,~~v\in \mathcal{V}_h^0,~\phi\in \mathcal{C}_h,
\end{equation}
\begin{equation}\label{eqFullyDiscreteVariationalPrinciples02}
-\( \Delta_{h}\phi,\varphi\)_h=\(\grad_h\phi,\grad_h\varphi \)_h,~~\phi,\varphi\in \mathcal{C}_h.
\end{equation}

 The following lemma is a direct consequence of Lemma 5.1 in \cite{kou2019EF}.
\begin{lem}\label{lemDiscreteLapalace}
Let   $\phi_-=\min(\phi-a,0)$ and $\phi_+=\max(\phi-b,0)$, where $\phi\in \mathcal{C}_h$ and $a<b$. Then we have
\begin{align}\label{eqDiscreteLapalace01A}
\|\grad_h \phi_{-} \|_h^2\leq-  \(\Delta_h \phi ,\phi_{-} \)_h  ,
\end{align}
\begin{align}\label{eqDiscreteLapalace01B}
\|\grad_h \phi_{+} \|_h^2\leq -  \(\Delta_h \phi ,\phi_{+} \)_h  .
\end{align}
\end{lem}
We note that for the homogeneous Dirichlet and periodic boundary conditions, we have the similar results to \eqref{eqDiscreteLapalace01A} and \eqref{eqDiscreteLapalace01B}. 

\section{Stabilized energy factorization approach and discrete scheme}\label{secSchemeFH}

 
Let $\phi^n$ denote the discrete function at the time level $t=t_n$, where $n\geq0$. 
The key  idea of  the energy factorization (EF) approach  is that we first   factorize the energy function  as follows
 \begin{equation}\label{eqEnergyFactorizationBasic01}
E(\phi) =  \int_\Omega \sum_{i=1}^M \Phi_{i}(\phi)\cdot\Psi_{i}(\phi)d\x,
\end{equation}
where $\Phi_{i}(\phi)$ and $\Psi_{i}(\phi)$ are the   energy factors, $M\geq1$, and then we derive the following energy inequality   using the properties  of the factors 
 \begin{equation}\label{eqEnergyFactorizationBasic02}
E(\phi^{n+1})-E(\phi^{n}) \leq  \int_\Omega \mu\(\phi^{n+1},\phi^n,\cdots,\phi^0)\\(\phi^{n+1}-\phi^n\)d\x.
\end{equation}
 Thus, $\mu$ is the  discrete general chemical potential, which is usually a functional of $\phi^{n+1}$ and $\phi^n$, and more generally may   rely  on $\phi^{n-1},\cdots,\phi^0$ as well. For a specific energy $E(\phi)$, there may be several different  factorization forms, but   not all of the resultant discrete   chemical potentials $\mu$  are   linear with respect to $\phi^{n+1}$. So an ingenious factorization approach is  required   for the sake of  obtaining a linear efficient scheme.
 
As a simple example,   we consider the gradient   free energy 
 \begin{align}\label{eqEnergyFactorizationGradEnergy01}
E_\grad(\phi)=\frac{1}{2}\int_\Omega \epsilon^2|\grad \phi|^2 d\x,
\end{align}
and we can derive that
\begin{align}\label{eqEnergyFactorizationGradEnergy01}
E_\grad(\phi^{n+1})-E_\grad(\phi^n)
&=\frac{1}{2}\int_\Omega\epsilon^2\(|\grad \phi^{n+1}|^2-|\grad \phi^{n}|^2\)d\x\nonumber\\
&=\int_\Omega \epsilon^2\grad \phi^{n+1}\cdot\grad\(\phi^{n+1}-\phi^{n}\)d\x
-\frac{1}{2}\int_\Omega\epsilon^2|\grad(\phi^{n+1}-\phi^n)|^2d\x\nonumber\\
&\leq\int_\Omega \epsilon^2\grad \phi^{n+1}\cdot\grad\(\phi^{n+1}-\phi^{n}\)d\x\nonumber\\
&=-\int_\Omega \(\phi^{n+1}-\phi^{n}\)\epsilon^2\Delta \phi^{n+1} d\x,
\end{align}
which leads to the classical  implicit chemical potential $\mu_\grad^{n+1}=-\epsilon^2\Delta \phi^{n+1}$.
The other treatment for $E_\grad(\phi)$  is described as follows
\begin{align}\label{eqEnergyFactorizationGradEnergy02}
E_\grad(\phi^{n+1})-E_\grad(\phi^n)
&=\frac{1}{2}\int_\Omega \epsilon^2\grad \(\phi^{n+1}+\phi^{n}\)\cdot\grad\(\phi^{n+1}-\phi^{n}\)d\x\nonumber\\
&=-\frac{1}{2}\int_\Omega \(\phi^{n+1}-\phi^{n}\)\epsilon^2\Delta \(\phi^{n+1}+\phi^{n}\) d\x,
\end{align}
which yields the classical  semi-implicit chemical potential $\mu_\grad^{n+1}=-\frac{1}{2}\epsilon^2\Delta \(\phi^{n+1}+\phi^{n}\)$.

\subsection{Stabilized energy factorization approach for Flory--Huggins   potential}
In order to treat  the  logarithmic Flory--Huggins   potential,  we introduce a stabilized energy factorization approach, which combines the stability technique with  the energy factorization approach proposed in \cite{kou2019EF} for the  PR-EOS based free energy. To apply the energy factorization approach, we define 
\begin{eqnarray}\label{eqLPfactorization01}
H_a(\phi)=\phi \ln\(\phi\),~~H_b(\phi)=\(1-\phi\)\ln\(1-\phi\),~~H_c(\phi)=\theta\phi\(1- \phi\),
\end{eqnarray}
and from this, we rewrite $F(\phi)$ as
\begin{eqnarray}\label{eqLPfactorization02}
F(\phi)=\(\lambda+1\)\(H_a(\phi) + H_b(\phi)\)+H_c(\phi)-\lambda\(H_a(\phi) + H_b(\phi)\),
\end{eqnarray}
where $\lambda\geq0$ is a constant. The  stability term  is crucial to ensure  the symmetric  positive definiteness and  discrete maximum principle of the resultant scheme in the case with a large $\theta$, which will be demonstrated in Subsection \ref{subsecExistenceFH} and  Subsection \ref{subsecMaximumprincipleFH} respectively.

It is clearly observed  that $H_a(\phi)$  as well as   $H_b(\phi)$  can be factorized into the product of a linear function    and a  logarithm function, while $H_c(\phi)$ is a product of two linear functions. 
 For $\phi^{n}>0$ and $\phi^{n+1}>0$, we can deduce that
 \begin{eqnarray}\label{eqLPfactorization03}
H_a(\phi^{n+1})-H_a(\phi^{n})&=&\phi^{n+1}\ln\(\phi^{n+1})-\phi^{n}\ln(\phi^{n}\)\nonumber\\
&=&\(\phi^{n+1}-\phi^{n}\)\ln(\phi^{n})+\phi^{n+1}\(\ln(\phi^{n+1})-\ln(\phi^{n})\)\nonumber\\
&\leq&\(\ln\(\phi^{n}\)+\frac{\phi^{n+1}}{\phi^{n}}\)\(\phi^{n+1}-\phi^{n}\),
\end{eqnarray}
where the last inequality is obtained  using the concavity of $\ln(\phi)$  as 
$$\ln(\phi^{n+1})\leq\ln(\phi^{n})+\frac{1}{\phi^{n}}\(\phi^{n+1}-\phi^{n}\).$$
Noting that $\ln(1-\phi)$ is also concave, we have
$$\ln(1-\phi^{n+1})\leq\ln(1-\phi^{n})+\frac{1}{\phi^{n}-1}\(\phi^{n+1}-\phi^{n}\).$$
Thus, for $\phi^{n}<1$ and $\phi^{n+1}<1$, we  deduce that
 \begin{eqnarray}\label{eqLPfactorization04}
H_b(\phi^{n+1})-H_b(\phi^{n})&=&\(1-\phi^{n+1}\)\ln\(1-\phi^{n+1})-\(1-\phi^{n}\)\ln(1-\phi^{n}\)\nonumber\\
&=&-\(\phi^{n+1}-\phi^{n}\)\ln(1-\phi^{n})\nonumber\\
&&+\(1-\phi^{n+1}\)\(\ln(1-\phi^{n+1})-\ln(1-\phi^{n})\)\nonumber\\
&\leq&\(-\ln(1-\phi^{n})+\frac{1-\phi^{n+1}}{\phi^{n}-1}\)\(\phi^{n+1}-\phi^{n}\).
\end{eqnarray}
For $H_c(\phi)$, we apply the factorization approach to deal with $H_c(\phi)$ as follows
 \begin{eqnarray}\label{eqLPfactorization05a}
H_c(\phi^{n+1})-H_c(\phi^{n})=\theta\(1-\phi^{n+1}-\phi^{n}\)\(\phi^{n+1}-\phi^{n}\).
\end{eqnarray}
Due to the concavity of $H_a(\phi)$, we have  $\ln(\phi^{n})-\ln(\phi^{n+1})\leq\frac{1}{\phi^{n+1}}\(\phi^{n}-\phi^{n+1}\)$, thus  we deduce 
\begin{eqnarray}\label{eqLPfactorization05b}
H_a(\phi^{n})-H_a(\phi^{n+1}) &=&\phi^{n}\ln\(\phi^{n}\)-\phi^{n+1}\ln\(\phi^{n+1}\)\nonumber\\
&=&\(\phi^{n}-\phi^{n+1}\)\ln(\phi^{n})+\phi^{n+1}\(\ln(\phi^{n})-\ln(\phi^{n+1})\)\nonumber\\
 &\leq&  -\(\ln\(\phi^{n}\)+1\)\(\phi^{n+1}-\phi^{n}\).
\end{eqnarray}
It is similar to deduce that
\begin{eqnarray}\label{eqLPfactorization05c}
H_b(\phi^{n})-H_b(\phi^{n+1})  \leq  \(\ln\(1-\phi^{n}\)+1\)\(\phi^{n+1}-\phi^{n}\).
\end{eqnarray}
Using  \eqref{eqLPfactorization03}-\eqref{eqLPfactorization05c}, we  define the discrete  chemical potential as
\begin{eqnarray}\label{eqLPfactorization05}
f(\phi^n,\phi^{n+1}) &=&\(\lambda+1\)\(\ln\(\phi^{n}\)+\frac{\phi^{n+1}}{\phi^{n}}-\ln(1-\phi^{n})-\frac{1-\phi^{n+1}}{1-\phi^{n}}\)\nonumber\\
&&+\theta\(1-\phi^{n+1}-\phi^{n}\)-\lambda\(\ln\(\phi^{n}\)+1\)+\lambda\(\ln\(1-\phi^{n}\)+1\)\nonumber\\
&=&\ln\(\phi^{n}\)-\ln(1-\phi^{n})+\(\lambda+1\)\(\frac{\phi^{n+1}}{\phi^{n}}-\frac{1-\phi^{n+1}}{1-\phi^{n}}\)\nonumber\\
&&+\theta\(1-\phi^{n+1}-\phi^{n}\).
\end{eqnarray}
From    \eqref{eqLPfactorization02}-\eqref{eqLPfactorization05}, we  obtain the following energy inequality    
\begin{align}\label{eqLPfactorization06}
F(\phi^{n+1})-F(\phi^{n})\leq f(\phi^n,\phi^{n+1})\(\phi^{n+1}-\phi^{n}\).
\end{align}

There is  the other alternative approach for $H_c(\phi)$ on the basis of the concavity of $H_c(\phi)$, which is expressed as  
 \begin{eqnarray}\label{eqLPfactorization05alternative01}
H_c(\phi^{n+1})-H_c(\phi^{n})\leq\theta\(1-2\phi^{n}\)\(\phi^{n+1}-\phi^{n}\).
\end{eqnarray}
In this case, the chemical potential $f(\phi^n,\phi^{n+1})$ is defined as below
\begin{eqnarray}\label{eqLPfactorization05alternative02}
f(\phi^n,\phi^{n+1}) &=&\ln\(\phi^{n}\)-\ln(1-\phi^{n})+\(\lambda+1\)\(\frac{\phi^{n+1}}{\phi^{n}}-\frac{1-\phi^{n+1}}{1-\phi^{n}}\)\nonumber\\
&&+\theta\(1-2\phi^{n}\).
\end{eqnarray}
It is observed from numerical tests that the treatment given in  \eqref{eqLPfactorization05a} performs  better than that of \eqref{eqLPfactorization05alternative01}.

 \subsection{Fully discrete    scheme}
The semi-implicit   discrete    scheme for  the Allen-Cahn equation with the logarithmic Flory--Huggins potential
 is stated as: given   $\phi^{n}\in \mathcal{C}_h$, find   $\phi^{n+1}\in \mathcal{C}_h$ such that
 \begin{equation}\label{eqFullyDiscreteSchmLP}
 \frac{ \phi^{n+1}-\phi^{n}}{\tau}
 -\epsilon^2\Delta_h\phi^{n+1}+f(\phi^n,\phi^{n+1})=0,
 \end{equation}
 where $\tau$ denotes the time step size, $f(\phi^n,\phi^{n+1})$ is  defined in \eqref{eqLPfactorization05} or \eqref{eqLPfactorization05alternative02} and $\phi^0$ is provided  by the initial condition. 
 
 In what follows, we will consider  only $f(\phi^n,\phi^{n+1})$ defined in \eqref{eqLPfactorization05} in  theoretical analysis, but the   scheme with \eqref{eqLPfactorization05alternative02} can be    analyzed in a very similar routine.

For convenience of theoretical analysis, using \eqref{eqLPfactorization05},  we rewrite \eqref{eqFullyDiscreteSchmLP} as the following equivalent from 
\begin{equation}\label{eqFullyDiscreteSchmLPEqt}
\frac{1}{\tau}\(\phi^{n+1}-\phi^{n}\)-\epsilon^2 \Delta_h \phi^{n+1}+\nu_{\theta,\lambda}(\phi^{n})\phi^{n+1}=r_{\theta,\lambda}(\phi^{n}),
\end{equation}
where  $\nu_{\theta,\lambda}(\phi)$ and $r_{\theta,\lambda}(\phi)$ are defined as follows
\begin{equation}\label{eqDiscreteEqnsLinTermLP}
\nu_{\theta,\lambda}(\phi)=\(\lambda+1\)\(\frac{1}{\phi}+\frac{1}{1-\phi}\)-\theta,
\end{equation}
\begin{align}\label{eqDiscreteEqnsSourceTermLP}
r_{\theta,\lambda}(\phi)&=-\ln\(\phi\)+\ln(1-\phi)+\frac{\lambda+1}{1-\phi}
-\theta\(1-\phi\).
\end{align}
\begin{rem}
Since the linear discrete chemical potentials are  obtained using the energy factorization approach,   the proposed discrete scheme is   linear,  easy to implement, and preserve the original energy dissipation law.    However, the commonly used convex splitting approach for  the  logarithmic Flory--Huggins   potential results in the nonlinear discrete chemical potential as well as the nonlinear schemes \cite{chen2017LP}.
\end{rem}
 
 \section{Theoretical analysis}\label{secanalysis}
In this section,  we will prove the well-posedness and maximum principle of the discrete solution.  In particular, we will   demonstrate that $\lambda$ is essential to ensure the maximum principle. The unconditional energy stability of the proposed scheme will be proved as well. 

\subsection{Existence   and uniqueness}\label{subsecExistenceFH}

For given $\phi^{n}$, we define the linear operator as follows
\begin{equation}\label{eqSPDoperatorLP}
\Q^n_{\theta,\lambda}=-\epsilon^2 \Delta_h +\nu_{\theta,\lambda}(\phi^{n}).
\end{equation}

\begin{lem}\label{lemSPDoperatorLP}
Assume  that $0<\phi^{n}< 1$. For given $\theta>2$, if $\lambda$ is chosen such that  
\begin{equation}\label{eqlambdaCondition}
\lambda\geq0 ~~\textnormal{and}~~  \lambda>\frac{1}{4}\theta-1,
\end{equation}
 then $\Q^n_{\theta,\lambda}$ is symmetric and positive definite.
\end{lem}
\begin{proof}
$\Q^n_{\theta,\lambda}$ is symmetric since 
for $\phi,\varphi\in  \mathcal{C}_h$, we have
\begin{align}\label{eqSPDoperatorLPproof01}
\(\Q^n_{\theta,\lambda}\phi,\varphi\)_h&=-\epsilon^2 \(\Delta_h\phi,\varphi\)_h +\(\nu_{\theta,\lambda}(\phi^{n})\phi,\varphi\)_h\nonumber\\
&=\epsilon^2 \(\grad_h\phi,\grad_h\varphi \)_h +\(\nu_{\theta,\lambda}(\phi^{n})\phi,\varphi\)_h\nonumber\\
&=-\epsilon^2 \(\Delta_h\varphi,\phi\)_h +\(\nu_{\theta,\lambda}(\phi^{n})\varphi,\phi\)_h=\(\Q^n_{\theta,\lambda}\varphi,\phi\)_h.
\end{align}
 We now prove that $\Q^n_{\theta,\lambda}$ is positive definite.
Let $\xi(s)=\frac{1}{s}+\frac{1}{1-s}$ for a scalar number $0< s< 1$, then the derivative of $\xi(s)$ is expressed as
\begin{equation*}\label{eqSPDoperatorLPproof02a}
\xi(s)=-\frac{1}{s^2}+\frac{1}{(1-s)^2}.
\end{equation*}
We observe that $\xi'(s)\leq0$ for $s\in(0,\frac{1}{2}]$ and $\xi'(s)\geq0$ for $s\in[\frac{1}{2},1)$, thus  
\begin{equation}\label{eqSPDoperatorLPproof02}
 \nu_{\theta,\lambda}(s)\geq4\(\lambda+1\)-\theta>0.
\end{equation}
For $\phi\in  \mathcal{C}_h$ and $\phi\neq0$, we deduce that
\begin{align}\label{eqSPDoperatorLPproof03}
\(\Q^n_{\theta,\lambda}\phi,\phi\)_h&= \epsilon^2 \(\grad_h\phi,\grad_h\phi \)_h +\(\nu_{\theta,\lambda}(\phi^{n})\phi,\phi\)_h\nonumber\\
&\geq \epsilon^2 \|\grad_h\phi\|^2_h +\(4\(\lambda+1\)-\theta\)\|\phi\|_h^2>0,
\end{align}
which yields the positive definiteness of $\Q^n_{\theta,\lambda}$.
\end{proof}

\begin{rem}
For $2<\theta<4$, the condition \eqref{eqlambdaCondition} holds if  $\lambda=0$. But  for $\theta\geq4$, a positive stability constant   satisfying  $\lambda>\frac{1}{4}\theta-1$ is required to ensure the  positive definiteness of
$\Q^n_{\theta,\lambda}$.  We also note that the linear system of the scheme with \eqref{eqLPfactorization05alternative02} is symmetric and positive definite for any $\theta>2$ and $\lambda\geq0$.
 \end{rem}
 
 We now prove the existence   and uniqueness of the discrete solution of \eqref{eqFullyDiscreteSchmLPEqt}.
\begin{thm}\label{thmExistenceOfSolutionLP}
Assume  that $0< \phi^{n}< 1$. For any time step size $\tau>0$, there exists a unique      $\phi^{n+1}$ to solve
\eqref{eqFullyDiscreteSchmLPEqt}   in $\mathcal{C}_h$ provided that $\lambda$ is chosen to satisfy  \eqref{eqlambdaCondition}. 
\end{thm}

\begin{proof}
It suffices to prove that the following homogeneous equation has a unique zero solution   in $\mathcal{C}_h$
\begin{equation}\label{eqExistenceOfSolutionLPproof01}
\frac{1}{\tau}\phi-\epsilon^2 \Delta_h \phi+\nu_{\theta,\lambda}(\phi^{n})\phi=0,
\end{equation}
which  can be rewritten as
\begin{equation}\label{eqExistenceOfSolutionLPproof02}
\frac{1}{\tau}\phi+\Q^n_{\theta,\lambda}\phi=0.
\end{equation}
As a direct consequence of the symmetric positive definiteness of $\Q^n_{\theta,\lambda}$,  there is a unique solution $\phi\equiv0$. 
\end{proof}

\subsection{Discrete maximum principle}\label{subsecMaximumprincipleFH}

The maximum principle of the discrete solution essentially relies on the stability constant for given energy parameter $\theta$, thus we need to employ the proper stability constant $\lambda$ to ensure this key property.  We denote 
\begin{equation}\label{eqMaximumPrincipleConditionLower}
  L(\theta,\lambda)=\min_{0< \phi< 1} r_{\theta,\lambda}(\phi ),
\end{equation}
\begin{equation}\label{eqMaximumPrincipleConditionUpper}
 U(\theta,\lambda)=\max_{0< \phi< 1}\big(r_{\theta,\lambda}(\phi)-\nu_{\theta,\lambda}(\phi)\big).
\end{equation}
In order to ensure the maximum principle of the proposed scheme,  we introduce the following condition: for given  $\theta>2$,  $\lambda$ shall be taken to satisfy \eqref{eqlambdaCondition} and
\begin{equation}\label{eqMaximumPrincipleCondition}
 L(\theta,\lambda)>0,~~U(\theta,\lambda)<0.
\end{equation}
\begin{rem}
  We illustrate the reasonability of the condition \eqref{eqMaximumPrincipleCondition} in Figure \ref{LogarithmicFloryHugginsACeqMaximumPrincipleConditions}. Here,   the values of $\lambda$ are taken as follows
 \begin{equation}\label{eqMaximumPrincipleConditionlambda}
\lambda=\left\{\begin{array}{ccc}
0, & 2<\theta\leq3, \\
1, & 3<\theta\leq4.5, \\
2, & 4.5<\theta\leq6.
\end{array}\right.
\end{equation}
It can be observed from Figure \ref{LogarithmicFloryHugginsACeqMaximumPrincipleConditions} that the condition \eqref{eqMaximumPrincipleCondition} naturally holds for $ 2<\theta\leq3$ without any stability term,  while it is still  true  for $ 3<\theta\leq6$ when we take the small values (1 and 2) for the stability constant $\lambda$ as in \eqref{eqMaximumPrincipleConditionlambda}. We have also checked some cases of $ \theta>6$ and found  that   the condition \eqref{eqMaximumPrincipleCondition} is true if $\lambda$ is properly chosen. We also emphasize that \eqref{eqMaximumPrincipleCondition} is a good guide to choose $\lambda$ since it is clearly independent of discrete solutions.  For the scheme with \eqref{eqLPfactorization05alternative02}, we have the similar condition for the choice of the stability constant. 
\end{rem}
 \begin{figure}
           \centering \subfigure[]{
            \begin{minipage}[b]{0.5\textwidth}
            \centering
             \includegraphics[width=0.95\textwidth,height=0.65\textwidth]{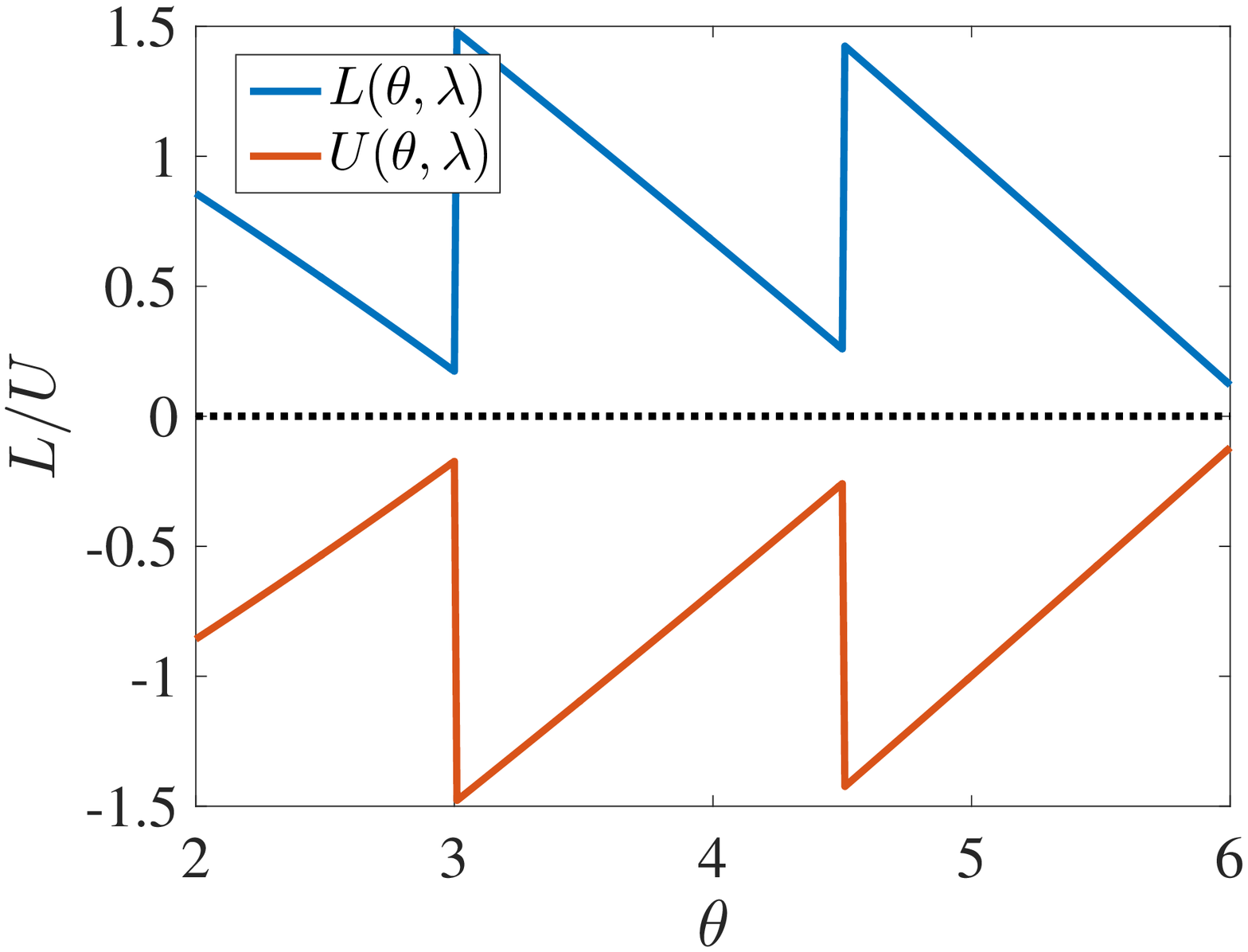}
             \label{LogarithmicFloryHugginsACeqMaximumPrincipleConditions}
            \end{minipage}
            }
            \caption{Verification of the condition \eqref{eqMaximumPrincipleCondition}.}
            \label{LogarithmicFloryHugginsACeqMaximumPrincipleConditions}
 \end{figure}

The following theorem demonstrates   that the scheme \eqref{eqFullyDiscreteSchmLPEqt} with the appropriate stability constant preserves the discrete maximum principle. 
\begin{thm}\label{thmMaximumPrincipleLP}
Assume  that $0< \phi^{0}<1,$ and $\lambda $ is chosen to satisfy \eqref{eqlambdaCondition} and \eqref{eqMaximumPrincipleCondition} for given $\theta>2$. For any time step size, the solutions of the scheme \eqref{eqFullyDiscreteSchmLPEqt} satisfy $0< \phi^{n}< 1$, $n\geq1$.
\end{thm}

\begin{proof}
By induction,  assuming $0< \phi^{n}< 1$, we prove $0< \phi^{n+1}< 1$.  There exists a unique $\phi^{n+1}\in \mathcal{C}_h$  in terms of Theorem \ref{thmExistenceOfSolutionLP}.  Let us define $\phi_{-}^{n+1}=\min(\phi^{n+1},0)$.  It is obtained from  \eqref{eqFullyDiscreteSchmLPEqt}   that
\begin{align}\label{thmMaximumPrincipleLPProof01}
&\frac{1}{\tau}\(\phi^{n+1}-\phi^{n},\phi_{-}^{n+1}\)_h-\epsilon^2 \(\Delta_h \phi^{n+1},\phi_{-}^{n+1}\)_h
+\(\nu_{\theta,\lambda}(\phi^{n})\phi^{n+1},\phi_-^{n+1}\)_h\nonumber\\
&~~~=\(r_{\theta,\lambda}(\phi^{n}),\phi_-^{n+1}\)_h.
\end{align}
The definition of $\phi_{-}^{n+1}$ yields
\begin{align}\label{thmMaximumPrincipleLPProof02}
\(\phi^{n+1},\phi_{-}^{n+1}\)_h=\|\phi_{-}^{n+1}\|_h^2.
\end{align}
Moreover, it is apparent that  $\phi_{-}^{n+1}\leq0$. Due to $\phi^{n}> 0$, we have
\begin{align}\label{thmMaximumPrincipleLPProof03}
\(\phi^{n},\phi_{-}^{n+1}\)_h\leq0.
\end{align}
Thus the first term on the left-hand side of \eqref{thmMaximumPrincipleLPProof01} is bounded as 
\begin{align}\label{thmMaximumPrincipleLPProof04}
\(\phi^{n+1}-\phi^{n},\phi_{-}^{n+1}\)_h=\|\phi_{-}^{n+1}\|_h^2-\(\phi^{n},\phi_{-}^{n+1}\)_h\geq\|\phi_{-}^{n+1}\|_h^2.
\end{align}
Applying \eqref{eqDiscreteLapalace01A}, the second term on the left-hand side of \eqref{thmMaximumPrincipleLPProof01} is bounded as 
\begin{align}\label{thmMaximumPrincipleLPProof05}
-\epsilon^2 \(\Delta_h \phi^{n+1},\phi_{-}^{n+1}\)_h\geq\epsilon^2\|\grad_h \phi_{-}^{n+1}\|_h^2.
\end{align}
The   third term  on the left-hand side of \eqref{thmMaximumPrincipleLPProof01}  is bounded below  
\begin{align}\label{thmMaximumPrincipleLPProof06}
\(\nu_{\theta,\lambda}(\phi^{n})\phi^{n+1},\phi_-^{n+1}\)_h\geq\(4\(\lambda+1\)-\theta\)\|\phi_-^{n+1}\|_h^2.
\end{align}
Using the condition \eqref{eqMaximumPrincipleCondition} and taking into account $\phi_{-}^{n+1}\leq0$,    the right-hand side of \eqref{thmMaximumPrincipleLPProof01}  is estimated   as
\begin{align}\label{thmMaximumPrincipleLPProof07}
\(r_{\theta,\lambda}(\phi^{n}),\phi_-^{n+1}\)_h
\leq\(L(\theta,\lambda),\phi_-^{n+1}\)_h  \leq0.
\end{align}
We combine \eqref{thmMaximumPrincipleLPProof02}-\eqref{thmMaximumPrincipleLPProof07} and obtain
\begin{align}\label{thmMaximumPrincipleLPProof08}
\frac{1}{\tau}\|\phi_{-}^{n+1}\|_h^2+\epsilon^2\|\grad_h \phi_{-}^{n+1}\|_h^2+\(4\(\lambda+1\)-\theta\)\|\phi_-^{n+1}\|_h^2
\leq0.
\end{align}
It follows from \eqref{thmMaximumPrincipleLPProof08} that
 $\|\phi_{-}^{n+1}\|_h^2=0$ and consequently, we get $\phi^{n+1}\geq 0$.
 We further prove that $\phi^{n+1}> 0$   by contradiction. Suppose that there exists at least a grid point  $(i_0+\frac{1}{2},j_0+\frac{1}{2})$  such that    $\phi^{n+1}_{i_0+\frac{1}{2},j_0+\frac{1}{2}}=0$, then the equation \eqref{eqFullyDiscreteSchmLPEqt} at  $(i_0+\frac{1}{2},j_0+\frac{1}{2})$ has the form
 \begin{align}\label{thmMaximumPrincipleLPProof08a}
-\frac{1}{\tau}\phi^{n}_{i_0+\frac{1}{2},j_0+\frac{1}{2}}-\epsilon^2 \Delta_h \phi_{i_0+\frac{1}{2},j_0+\frac{1}{2}}^{n+1}=r_{\theta,\lambda}\big(\phi^{n}_{i_0+\frac{1}{2},j_0+\frac{1}{2}}\big).
\end{align}
The right-hand side of  \eqref{thmMaximumPrincipleLPProof08a} can be estimated using \eqref{eqMaximumPrincipleCondition} as 
 \begin{align}\label{thmMaximumPrincipleLPProof08b}
r_{\theta,\lambda}\big(\phi^{n}_{i_0+\frac{1}{2},j_0+\frac{1}{2}}\big)  \geq  L(\theta,\lambda)>0.
\end{align}
For the two terms on the left-hand side of  \eqref{thmMaximumPrincipleLPProof08a}, we obtain 
 \begin{equation}\label{thmMaximumPrincipleLPProof08c}
-\frac{1}{\tau}\phi^{n}_{i_0+\frac{1}{2},j_0+\frac{1}{2}}-\epsilon^2 \Delta_h \phi_{i_0+\frac{1}{2},j_0+\frac{1}{2}}^{n+1} <0,
\end{equation}
which contradict  \eqref{thmMaximumPrincipleLPProof08b}.  Therefore, $\phi^{n+1}> 0$.

To prove $\phi^{n+1}< 1$, we define $\phi_{+}^{n+1}=\max(\phi^{n+1}-1,0)$.
We get from \eqref{eqFullyDiscreteSchmLPEqt} that
\begin{align}\label{thmMaximumPrincipleLPProof09}
&\frac{1}{\tau}\(\phi^{n+1}-\phi^{n},\phi_{+}^{n+1}\)_h-\epsilon^2 \(\Delta_h \phi^{n+1},\phi_{+}^{n+1}\)_h
+\(\nu_{\theta,\lambda}(\phi^{n})(\phi^{n+1}-1),\phi_+^{n+1}\)_h\nonumber\\
&=\(r_{\theta,\lambda}(\phi^n)-\nu_{\theta,\lambda}(\phi^n), \phi_+^{n+1}\)_h.
\end{align}
The definition of $\phi_{+}^{n+1}$ implies  that
\begin{align}\label{thmMaximumPrincipleLPProof10}
\phi_{+}^{n+1}\geq0,~~\(\phi^{n+1}-1,\phi_{+}^{n+1}\)_h=\|\phi_{+}^{n+1}\|_h^2.
\end{align}
Taking into account $\phi^{n}< 1$, we have
\begin{align}\label{thmMaximumPrincipleLPProof11}
\(\phi^{n}-1,\phi_{+}^{n+1}\)_h\leq0.
\end{align}
 The first term of \eqref{thmMaximumPrincipleLPProof09} is bounded below using \eqref{thmMaximumPrincipleLPProof10} and \eqref{thmMaximumPrincipleLPProof11}
\begin{align}\label{thmMaximumPrincipleLPProof12}
\(\phi^{n+1}-\phi^{n},\phi_{+}^{n+1}\)_h=\|\phi_{+}^{n+1}\|_h^2-\(\phi^{n}-1,\phi_{+}^{n+1}\)_h\geq\|\phi_{+}^{n+1}\|_h^2.
\end{align}
Similar to \eqref{thmMaximumPrincipleLPProof05} and \eqref{thmMaximumPrincipleLPProof06}, we have
\begin{align}\label{thmMaximumPrincipleLPProof13}
-\epsilon^2 \(\Delta_h \phi^{n+1},\phi_{+}^{n+1}\)_h\geq\epsilon^2\|\grad_h \phi_{+}^{n+1}\|_h^2,
\end{align}
\begin{align}\label{thmMaximumPrincipleLPProof14}
\(\nu_{\theta,\lambda}(\phi^{n})(\phi^{n+1}-1),\phi_+^{n+1}\)_h \geq \(4\(\lambda+1\)-\theta\)\|\phi_+^{n+1}\|_h^2,
\end{align}
where \eqref{thmMaximumPrincipleLPProof13} results from \eqref{eqDiscreteLapalace01B}.
Applying the condition \eqref{eqMaximumPrincipleCondition} and   $\phi_{+}^{n+1}\geq0$, we derive
\begin{align}\label{thmMaximumPrincipleLPProof15}
\(r_{\theta,\lambda}(\phi^n)-\nu_{\theta,\lambda}(\phi^n), \phi_+^{n+1}\)_h
\leq\(U(\theta,\lambda),\phi_+^{n+1}\)_h \leq 0.
\end{align}
Finally, we reach 
\begin{align}\label{thmMaximumPrincipleLPProof16}
\frac{1}{\tau}\|\phi_{+}^{n+1}\|_h^2+\epsilon^2\|\grad_h \phi_{+}^{n+1}\|_h^2+\(4\(\lambda+1\)-\theta\)\|\phi_+^{n+1}\|_h^2
\leq0,
\end{align}
which implies that
 $\|\phi_{+}^{n+1}\|_h^2=0$, and thus,  $\phi^{n+1}\leq 1$.  We further prove that $\phi^{n+1}\neq 1$ in any grid cell  by contradiction. Suppose that      $\phi^{n+1}_{i_0+\frac{1}{2},j_0+\frac{1}{2}}=1$, and then the equation \eqref{eqFullyDiscreteSchmLPEqt} at  $(i_0+\frac{1}{2},j_0+\frac{1}{2})$ becomes
 \begin{equation}\label{thmMaximumPrincipleLPProof17}
\frac{1}{\tau}\(1-\phi^{n}_{i_0+\frac{1}{2},j_0+\frac{1}{2}}\)-\epsilon^2 \Delta_h \phi_{i_0+\frac{1}{2},j_0+\frac{1}{2}}^{n+1} =r_{\theta,\lambda}\big(\phi^{n}_{i_0+\frac{1}{2},j_0+\frac{1}{2}}\big)-\nu_{\theta,\lambda}\big(\phi^{n}_{i_0+\frac{1}{2},j_0+\frac{1}{2}}\big).
\end{equation}
The right-hand side of  \eqref{thmMaximumPrincipleLPProof17} can be estimated using \eqref{eqMaximumPrincipleCondition} as 
 \begin{align}\label{thmMaximumPrincipleLPProof18}
r_{\theta,\lambda}\big(\phi^{n}_{i_0+\frac{1}{2},j_0+\frac{1}{2}}\big)-\nu_{\theta,\lambda}\big(\phi^{n}_{i_0+\frac{1}{2},j_0+\frac{1}{2}}\big)  \leq  U(\theta,\lambda)<0,
\end{align}
while the two terms on the left-hand side of  \eqref{thmMaximumPrincipleLPProof17}  are estimated as 
 \begin{equation}\label{thmMaximumPrincipleLPProof19}
\frac{1}{\tau}\(1-\phi^{n}_{i_0+\frac{1}{2},j_0+\frac{1}{2}}\)>0,~~-\epsilon^2 \Delta_h \phi_{i_0+\frac{1}{2},j_0+\frac{1}{2}}^{n+1} \geq0.
\end{equation}
Substituting \eqref{thmMaximumPrincipleLPProof18} and \eqref{thmMaximumPrincipleLPProof19} into \eqref{thmMaximumPrincipleLPProof17} yields a contradiction, and consequently, we obtain $\phi^{n+1}<1$.
\end{proof}

\subsection{Unconditional energy stability}

For $\phi\in \mathcal{C}_h$, the discrete total free energy is defined as
\begin{equation}\label{eqFullydiscreteEnergyPL}
E_h(\phi)=\( F(\phi),1\)_h+\frac{1}{2}\epsilon^2\|\grad_h\phi\|_h^2.
\end{equation}
\begin{thm}\label{thmFullySchmEnergyStabilityPL}
Assume  that  $0< \phi^{0}<1$ and $\lambda $ is chosen to satisfy \eqref{eqlambdaCondition} and \eqref{eqMaximumPrincipleCondition} for given $\theta>2$. For any time step size $\tau$, the total free energy of the scheme  \eqref{eqFullyDiscreteSchmLPEqt}  is dissipated with time steps, i.e.,  
\begin{equation}\label{eqFullySchmEnergyStabilityPL}
E_h(\phi^{n+1})\leq E_h(\phi^n).
\end{equation}
\end{thm}
\begin{proof}
  Theorems \ref{thmExistenceOfSolutionLP} and \ref{thmMaximumPrincipleLP} argue  that there exists a unique $\phi^{n}\in\mathcal{C}_h$ and  $0< \phi^{n}<1$ for $n\geq1$.  The inequality \eqref{eqLPfactorization06}  yields 
\begin{align}\label{eqFullySchmEnergyStabilityPLProof01}
\(F(\phi^{n+1})-F(\phi^{n}),1\)_h\leq \(f(\phi^{n},\phi^{n+1}),\phi^{n+1}-\phi^n\)_h.
\end{align}
On the other hand, using \eqref{eqFullyDiscreteVariationalPrinciples02},   we can derive that
\begin{align}\label{eqFullySchmEnergyStabilityPLProof02}
\frac{1}{2}\(\|\grad_h \phi^{n+1}\|_h^2-\|\grad_h \phi^{n}\|_h^2\)
&=\(\grad_h \phi^{n+1}, \grad_h \(\phi^{n+1}-\phi^n\)\)_h-\frac{1}{2}\left\|\grad_h \(\phi^{n+1}-\phi^n\)\right\|_h^2\nonumber\\
&\leq\(\grad_h \phi^{n+1}, \grad_h \(\phi^{n+1}-\phi^n\)\)_h\nonumber\\
&=-\(\Delta_h \phi^{n+1},  \phi^{n+1}-\phi^n\)_h.
\end{align}
Using \eqref{eqFullyDiscreteSchmLP}, \eqref{eqFullySchmEnergyStabilityPLProof01} and \eqref{eqFullySchmEnergyStabilityPLProof02}, we deduce   
\begin{align}\label{eqFullySchmEnergyStabilityPLProof03}
E_h(\phi^{n+1})- E_h(\phi^n)
&\leq \(-\epsilon^2\Delta_h \phi^{n+1}+f(\phi^{n},\phi^{n+1}),\phi^{n+1}-\phi^n\)_h\nonumber\\
&\leq-\frac{1}{\tau}\|\phi^{n+1}-\phi^{n}\|_h^2,
\end{align}
which yields \eqref{eqFullySchmEnergyStabilityPL}.
\end{proof}

\begin{rem}
We note that  \eqref{eqMaximumPrincipleCondition} is just a sufficient condition to ensure the discrete maximum principle.
The  energy   dissipation law always holds  as long as $0<\phi^n<1$, $n\geq1$.  We apply the discrete chemical potential \eqref{eqLPfactorization05} for solving  the Cahn--Hilliard equation and find that $\phi^n$ always falls in  $(0,1)$,  thus the energy   dissipation is still validated. 
\end{rem}

\section{Numerical results}\label{secNumer}

In this section, we present numerical results to show the effectiveness of  the proposed scheme. In all numerical experiments, the computational domain is $\Omega=[-1,1]^2$, which is divided using a uniform mesh with $100\times100$ elements. Here, we take $\epsilon=0.05$. The proposed scheme admits   a very time step size in theory,  so we take the time step size $\tau=10^{10}$ for the purpose of verifying  this feature.

In this example,  we compare the discrete chemical potentials proposed in \eqref{eqLPfactorization05} and \eqref{eqLPfactorization05alternative02}.   We take $\theta=3$ and $\lambda=0$.  The initial condition of the phase variable is given as
 \begin{equation}\label{eqInitalConditionLP}
\phi(x,y)=\left\{\begin{array}{ccc}
10^{-5}, & |x|\leq0.35 ~~\textnormal{and}~~|y|\leq0.35, \\
1-10^{-5}, & \textnormal{otherwise}.
\end{array}\right.
\end{equation}

Figure \ref{LFHSquareACSchemeLFH12LFHLFHPara3StabConst0MaxMinPhiVals} depicts that the maximum and minimum values of the phase variable computed by the proposed scheme with \eqref{eqLPfactorization05} or \eqref{eqLPfactorization05alternative02} at each time point. Although   the initial   maximum and minimum values of $\phi^n$ are set to approach  the singular points,   all of $\phi^n$ never go  beyond  $(0,1)$ in the subsequent time, and consequently  the maximum principle is always  preserved by both \eqref{eqLPfactorization05} and \eqref{eqLPfactorization05alternative02}. 

 Figure \ref{LFHSquareLFHSquareACLFHLFHPara3StabConst0freeEnergyTwoSchemes} demonstrates that the total energies are always decreasing  with time steps.  The proposed scheme  with \eqref{eqLPfactorization05} has more rapid dissipation rate than that with  \eqref{eqLPfactorization05alternative02}, thus the square-shape phase simulated by the former can be    shrinking into a circle more rapidly. This phenomenon is observed  in Figure \ref{LFHSquareLFHSquareACTwoSchemesLFHLFHPara3StabConst0PhaseVariable}, which illustrates
  the   dynamical evolution process of the phase variable at different time steps.   The superiority of  \eqref{eqLPfactorization05}  results from the semi-implicit treatment for $H_c(\phi)$, while   it is treated   explicitly in \eqref{eqLPfactorization05alternative02}.

 \begin{figure}
           \centering \subfigure[the scheme with \eqref{eqLPfactorization05} ]{
            \begin{minipage}[b]{0.38\textwidth}
            \centering
             \includegraphics[width=0.95\textwidth,height=1.33in]{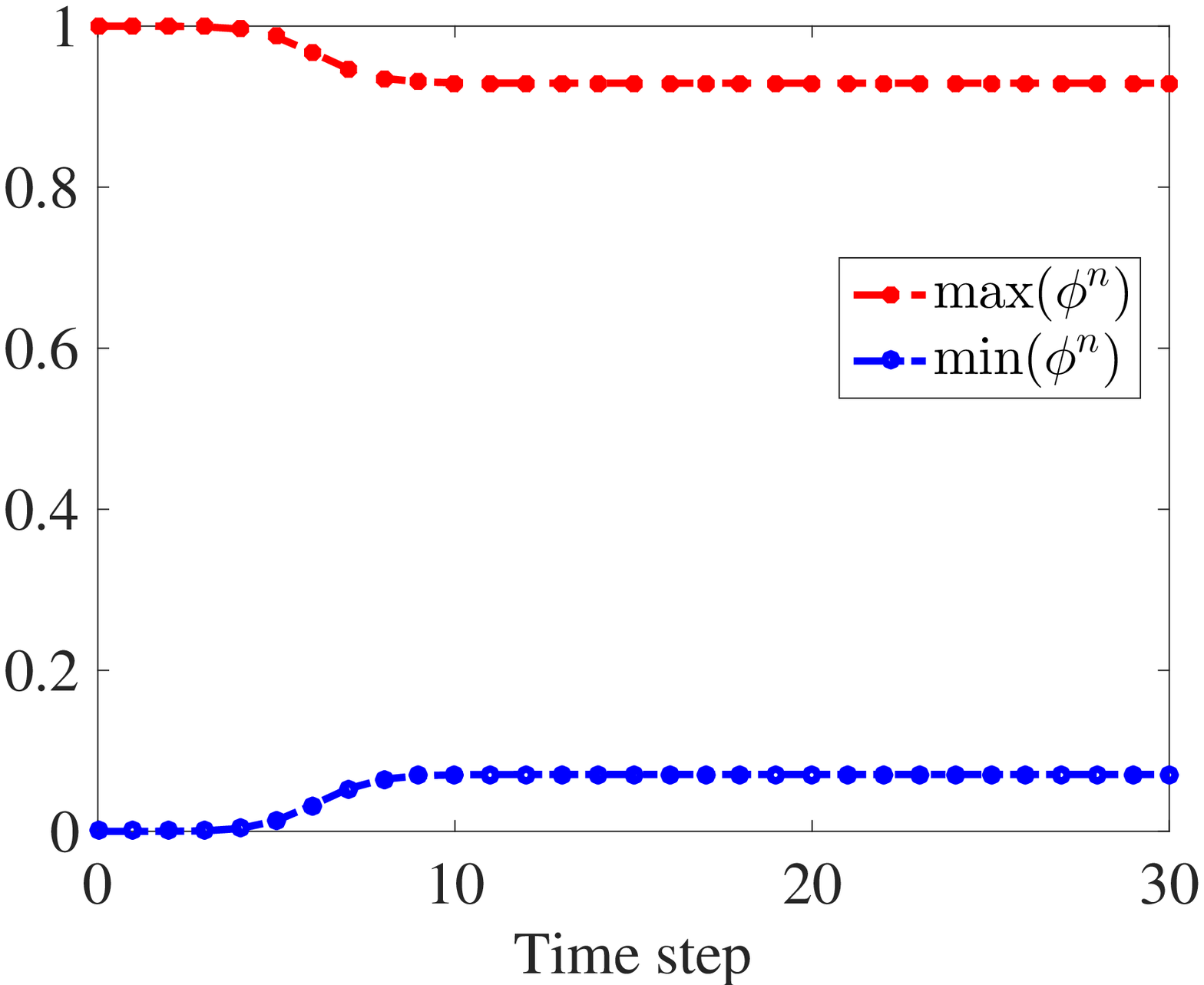}
             \label{LFHSquareACSchemeLFH2LFHLFHPara3StabConst0MaxMinPhiVals}
            \end{minipage}
            }
            \centering \subfigure[the scheme with  \eqref{eqLPfactorization05alternative02} ]{
            \begin{minipage}[b]{0.38\textwidth}
            \centering
             \includegraphics[width=0.95\textwidth,height=1.33in]{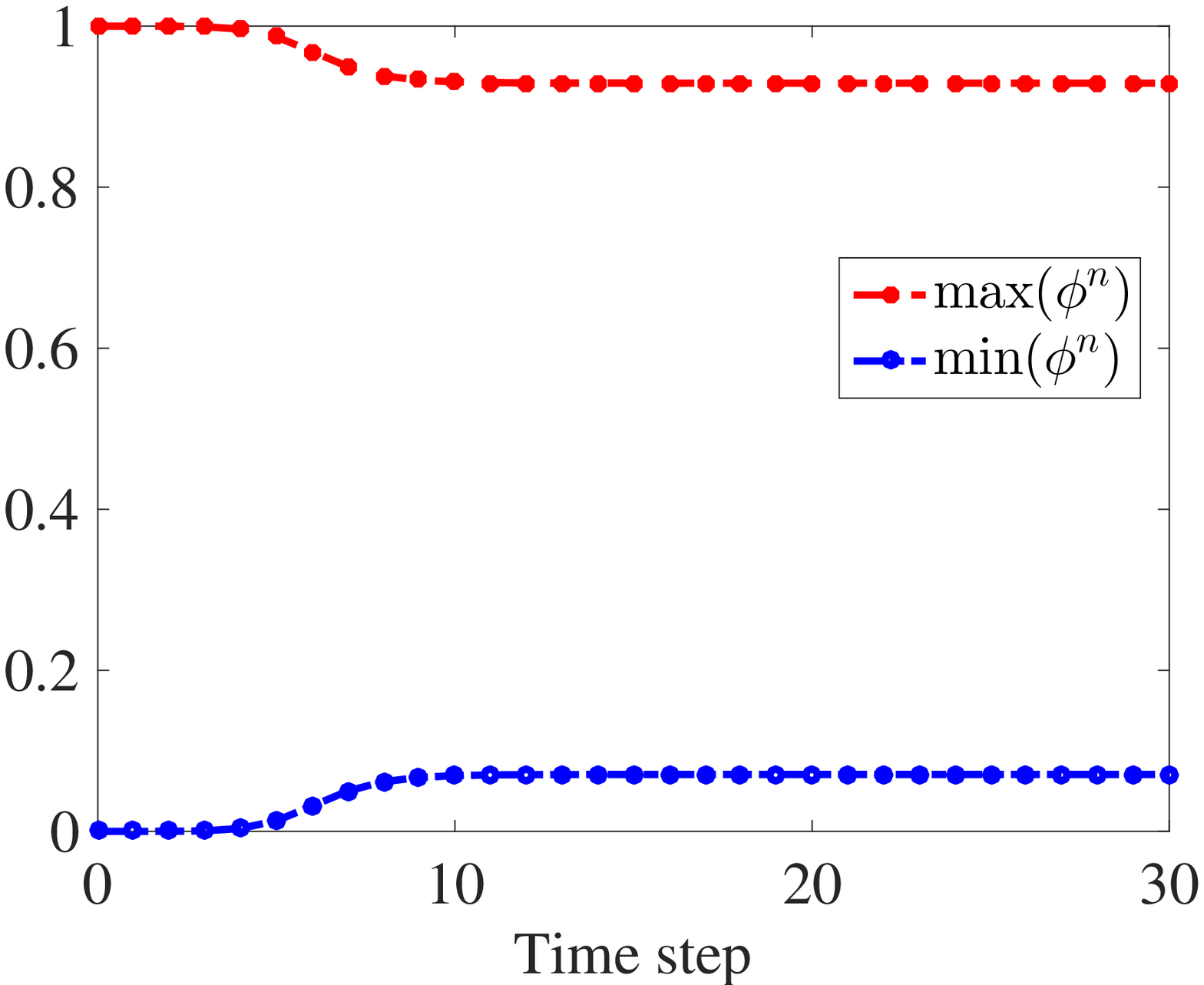}
             \label{LFHSquareACSchemeLFH1LFHLFHPara3StabConst0MaxMinPhiVals}
            \end{minipage}
            }
           \caption{Verification of the maximum principle.}
            \label{LFHSquareACSchemeLFH12LFHLFHPara3StabConst0MaxMinPhiVals}
 \end{figure}

\begin{figure}
           \centering \subfigure[]{
            \begin{minipage}[b]{0.5\textwidth}
            \centering
             \includegraphics[width=0.95\textwidth,height=0.6\textwidth]{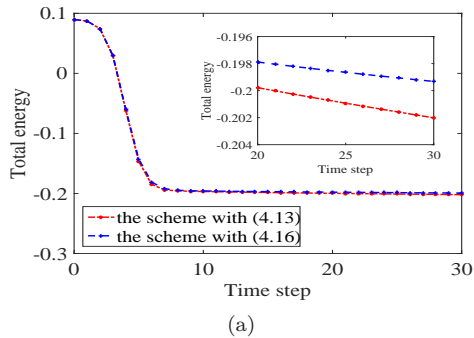}
             \label{LFHSquareLFHSquareACLFHLFHPara3StabConst0freeEnergyTwoSchemes}
            \end{minipage}
            }
           \caption{Energy curves.}
            \label{LFHSquareLFHSquareACLFHLFHPara3StabConst0freeEnergyTwoSchemes}
 \end{figure}

\begin{figure}
            \centering \subfigure[scheme with  \eqref{eqLPfactorization05},  $n=10$]{
            \begin{minipage}[b]{0.31\textwidth}
               \centering
             \includegraphics[width=0.95\textwidth,height=0.8\textwidth]{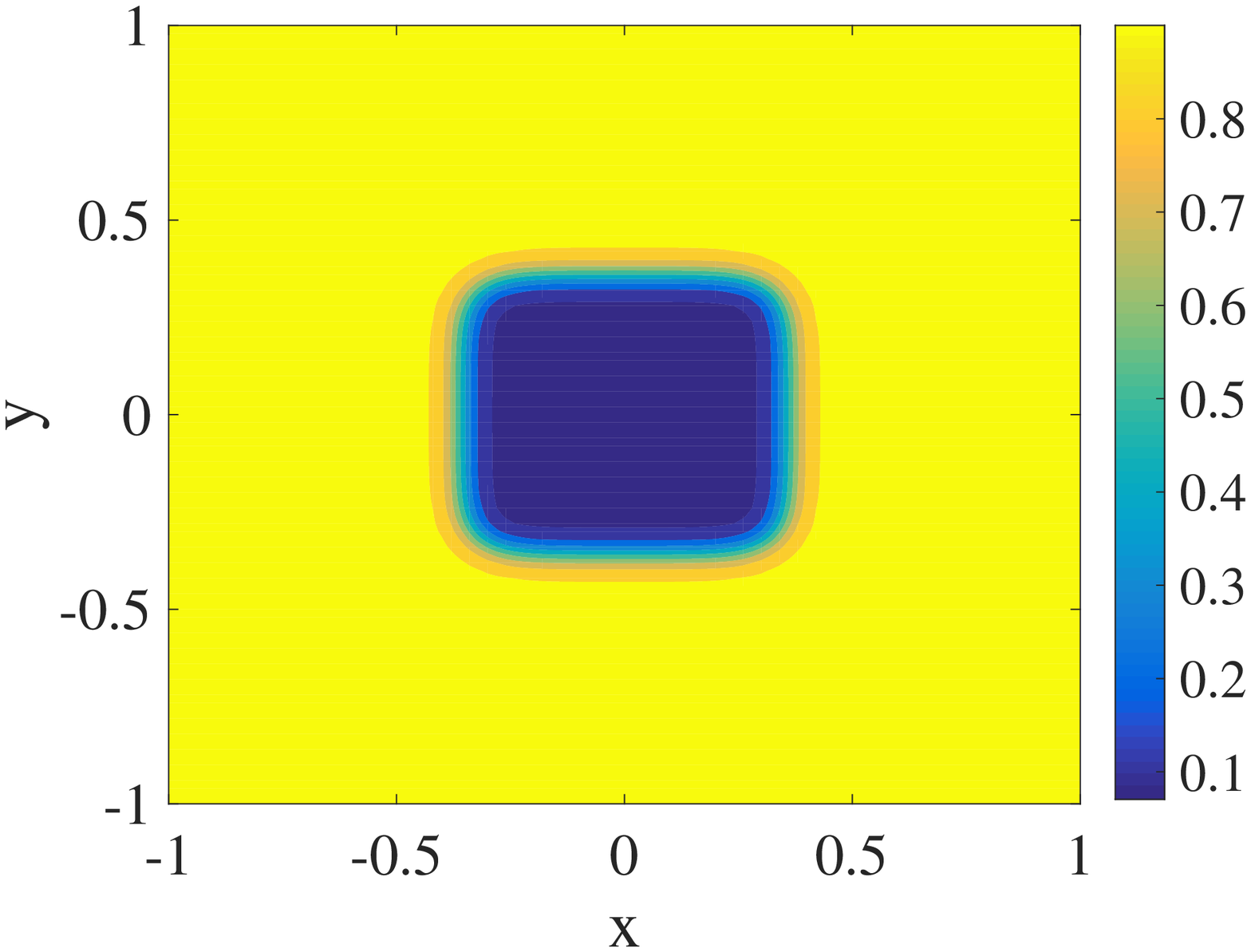}
              \label{LFHSquareLFHSquareACSchemeLFH2LFHStabConst0PhaseVariableOfiT0}
            \end{minipage}
            }
            \centering \subfigure[scheme with  \eqref{eqLPfactorization05},  $n=20$]{
            \begin{minipage}[b]{0.31\textwidth}
            \centering
             \includegraphics[width=0.95\textwidth,height=0.8\textwidth]{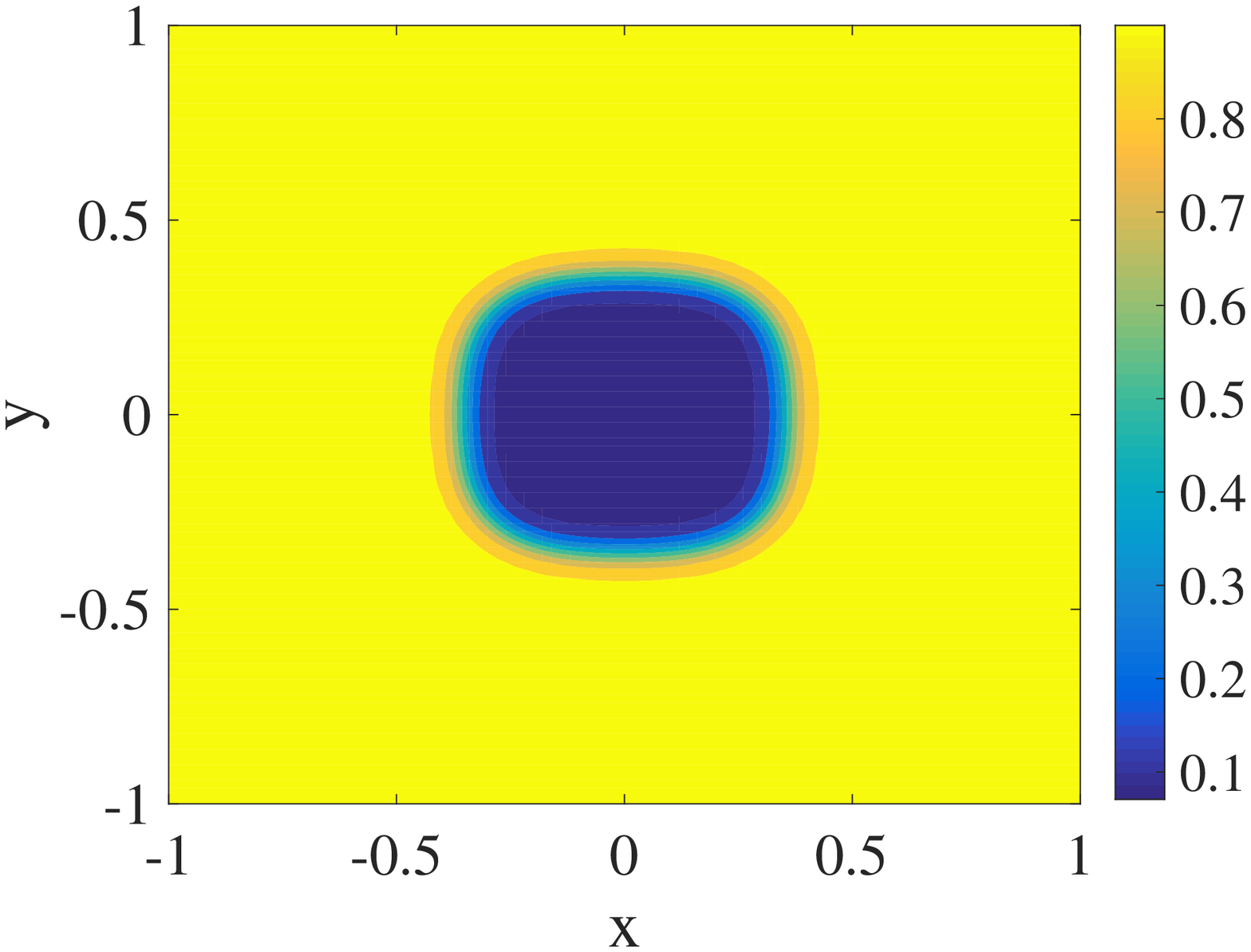}
            \end{minipage}
            }
           \centering \subfigure[scheme with  \eqref{eqLPfactorization05},  $n=30$]{
            \begin{minipage}[b]{0.3\textwidth}
            \centering
             \includegraphics[width=0.95\textwidth,height=0.8\textwidth]{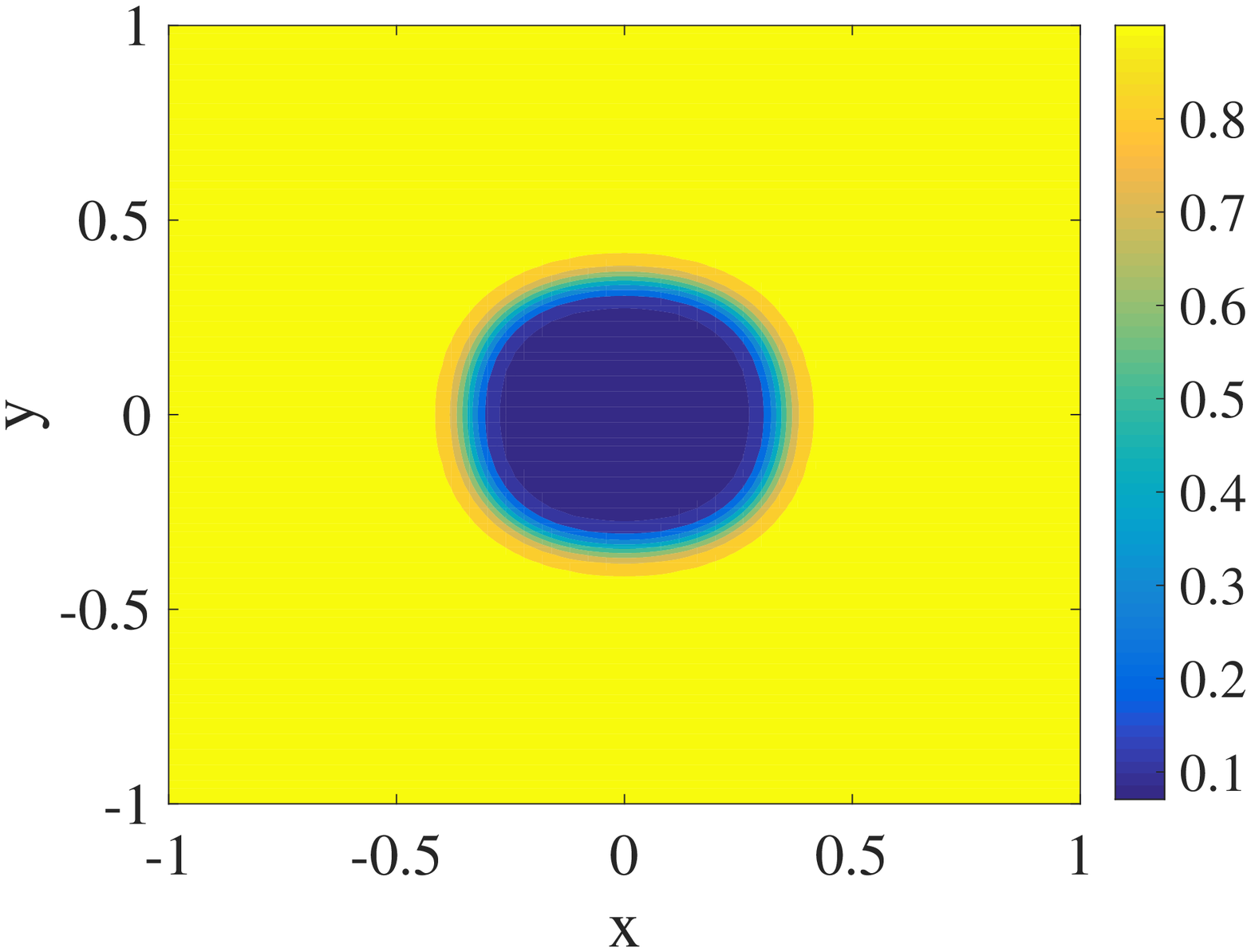}
            \end{minipage}
            }
           \centering \subfigure[scheme with  \eqref{eqLPfactorization05alternative02}, $n=10$]{
            \begin{minipage}[b]{0.31\textwidth}
               \centering
             \includegraphics[width=0.95\textwidth,height=0.8\textwidth]{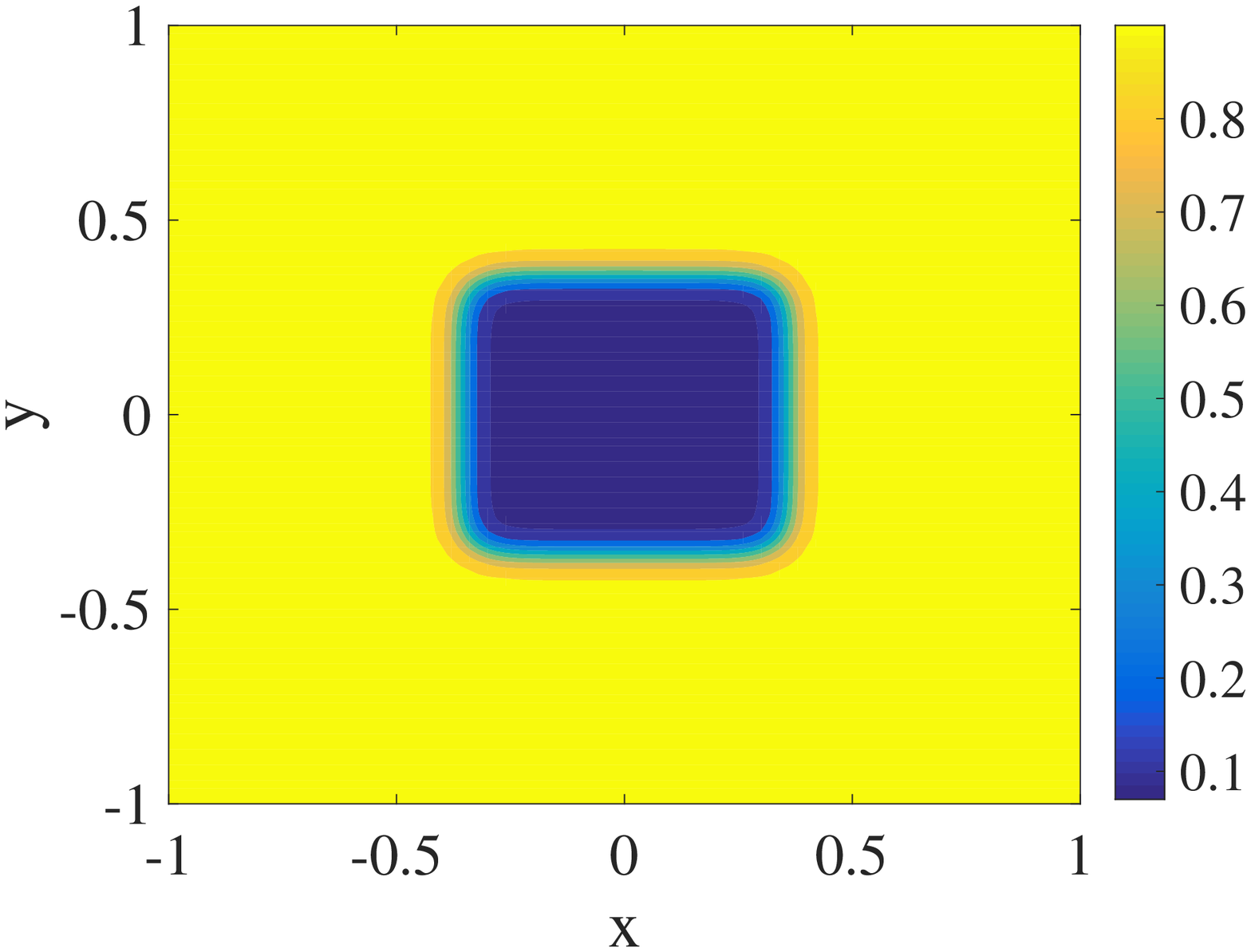}
              \label{LFHSquareLFHSquareACSchemeLFH2LFHStabConst0PhaseVariableOfiT0}
            \end{minipage}
            }
            \centering \subfigure[scheme with  \eqref{eqLPfactorization05alternative02}, $n=20$]{
            \begin{minipage}[b]{0.31\textwidth}
            \centering
             \includegraphics[width=0.95\textwidth,height=0.8\textwidth]{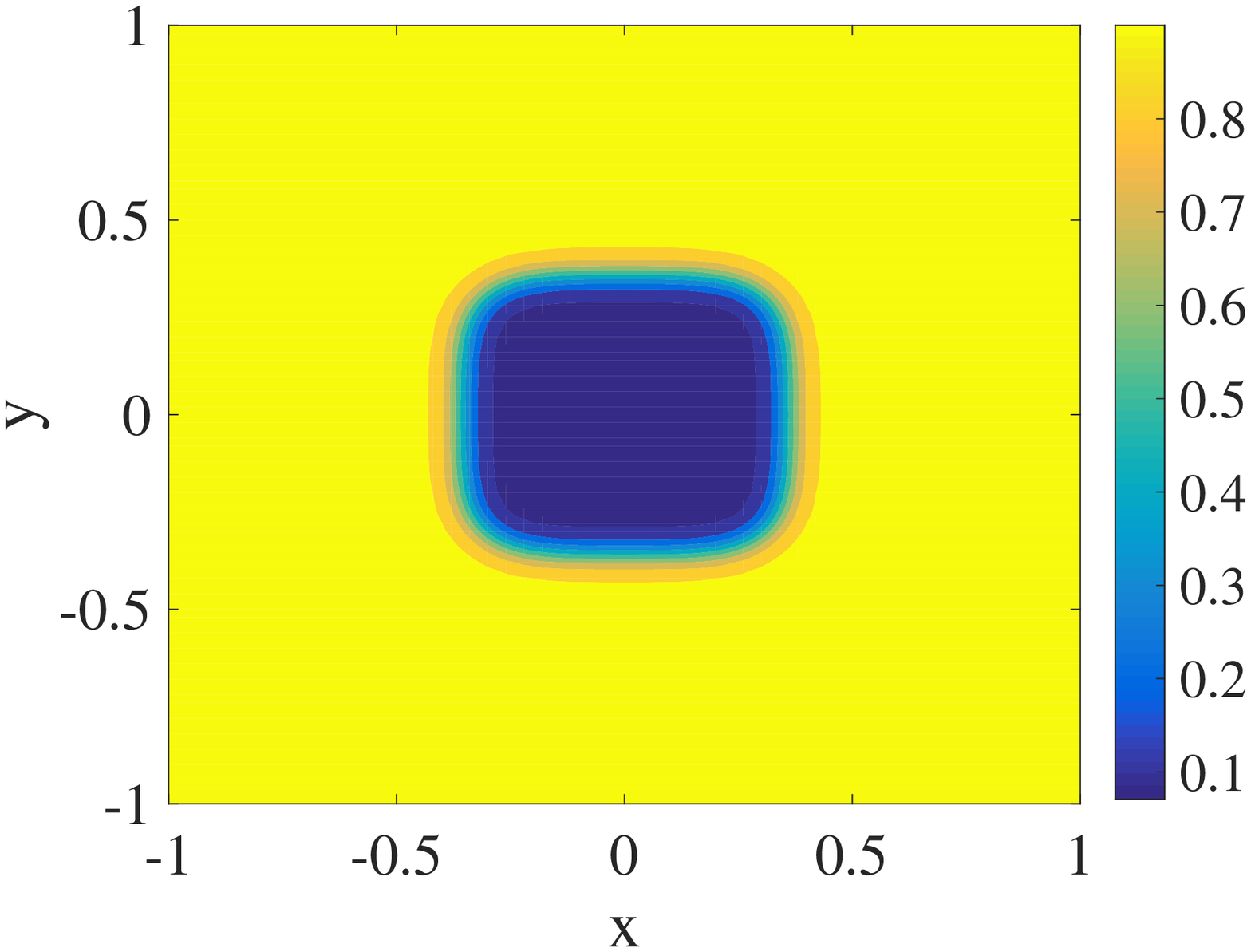}
            \end{minipage}
            }
           \centering \subfigure[scheme with  \eqref{eqLPfactorization05alternative02}, $n=30$]{
            \begin{minipage}[b]{0.3\textwidth}
            \centering
             \includegraphics[width=0.95\textwidth,height=0.8\textwidth]{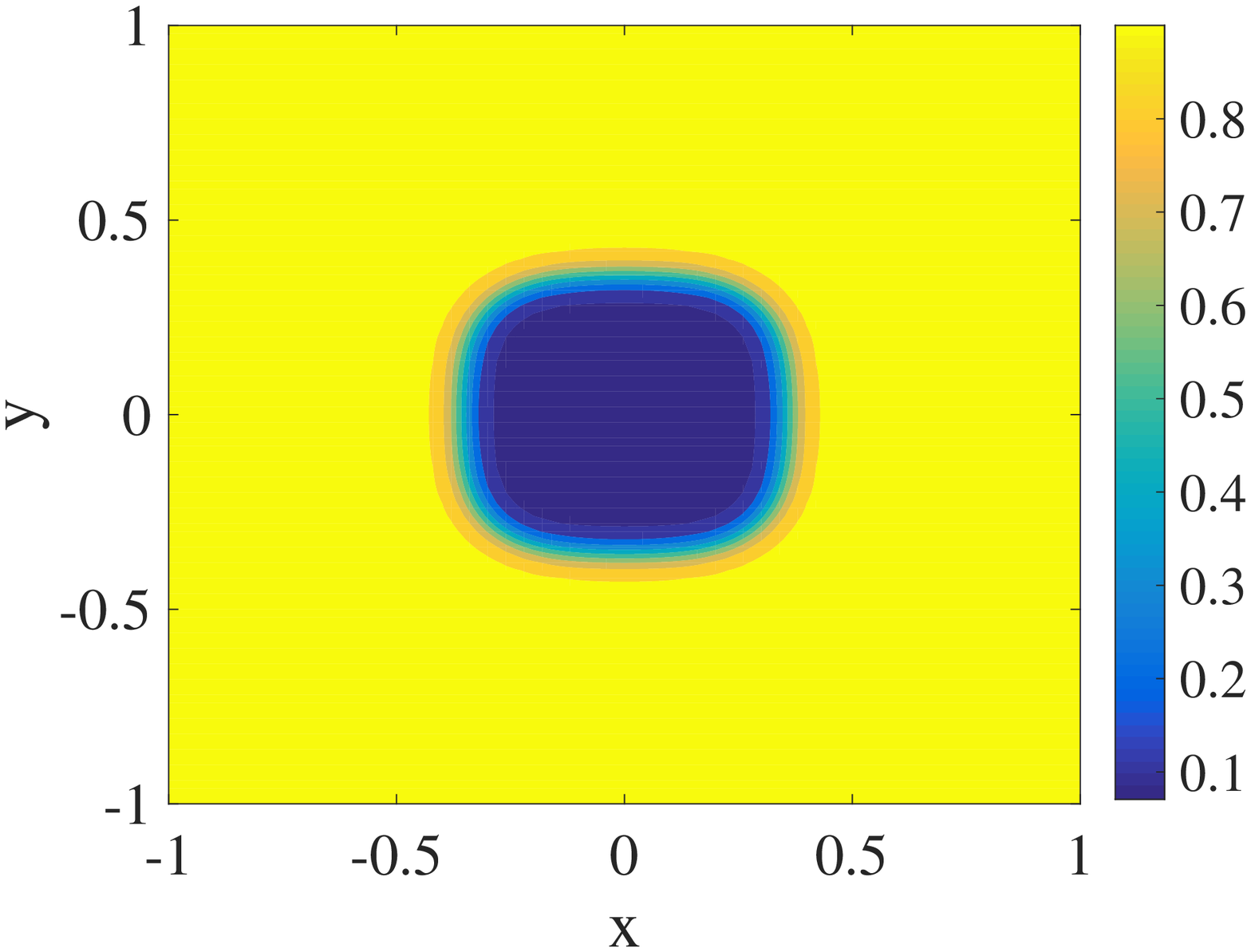}
            \end{minipage}
            }
           \caption{The dynamical evolution of the phase variable computed by the scheme with  \eqref{eqLPfactorization05} or \eqref{eqLPfactorization05alternative02} at different time steps.}
            \label{LFHSquareLFHSquareACTwoSchemesLFHLFHPara3StabConst0PhaseVariable}
 \end{figure}

\section{Conclusions}\label{secConclusions}


The stabilized  energy factorization  approach  has been developed  to treat the logarithmic Flory--Huggins   potential    semi-implicitly. The stability term can eliminate  the instability caused by the large energy parameter.     Compared to the prevalent convex-splitting approach  and auxiliary variable approaches, such approach leads to the simply  linear numerical scheme inheriting the original  energy dissipation law. To our best   knowledge,  the proposed scheme  is the first  such linear, unconditionally original  energy  stable scheme for the logarithmic Flory--Huggins   potential.  Moreover, the proposed scheme is rigorously proved to satisfy the discrete maximum principle under the appropriate choices of the stability constant. 
Numerical results  are in good agreement with theoretical analysis. 

In addition, the proposed scheme can be extended to the Cahn--Hilliard equation and the corresponding numerical results (not presented here) demonstrate that the maximum principle and unconditional energy stability  are still validated as long as the proper stability constant is taken although   theoretical proof  may be not available for   the maximum principle of the Cahn--Hilliard equation.

\small

\end{document}